\documentclass[12pt]{article}

\RequirePackage[OT1]{fontenc}
\RequirePackage{amsthm,amssymb,amsfonts,graphicx,subfigure,caption,bm,fullpage,wrapfig,setspace}
\RequirePackage[colorlinks,citecolor=blue,urlcolor=blue]{hyperref}
\usepackage{amsmath,amssymb,hyperref,enumerate}

\usepackage{algpseudocode,algorithm}

\numberwithin{equation}{section}
\theoremstyle{plain}
                          \newtheorem{thm}{Theorem}
                          \newtheorem{lem}{Lemma}
                          \newtheorem{prop}{Proposition}
\theoremstyle{remark}         
\theoremstyle{definition} \newtheorem{defn}{Definition}

\newcommand\om{\Omega}

\newcommand\A{\mathcal{A}}

\newcommand\F{\mathcal{F}}

\newcommand\bn{\bm{n}}

\newcommand\bi{\begin{itemize}}
\newcommand\ei{\end{itemize}}

\newcommand{\iid}{\stackrel{\mathrm{iid}}{\sim}}

\newcommand\independent{\protect\mathpalette{\protect\independenT}{\perp}}
\def\independenT#1#2{\mathrel{\rlap{$#1#2$}\mkern2mu{#1#2}}}
\newcommand\independentk{\independent_{\!\!\!k_1,k_2}}

\def\I{{\mathbf{1}}}

\newcommand{\beginsupplement}{%
        \setcounter{table}{0}
        \renewcommand{\thetable}{S\arabic{table}}%
        \setcounter{figure}{0}
        \renewcommand{\thefigure}{S\arabic{figure}}%
}

\def\woMR#1{\w@MR#1MR#1MR\relax}%
\def\w@MR#1MR#2MR#3\relax{#2}

\def\@MR#1 #2\relax#3{%
 \href{http://www.ams.org/mathscinet-getitem?mr=#1}%
 {\MRfixed{#3}}}%

\def\MRfixed{MR\woMR}%

\usepackage{natbib}

\usepackage{varioref}		
\labelformat{section}{Section~#1}
\labelformat{figure}{Figure~#1}
\labelformat{table}{Table~#1}	

\title{Fisher exact scanning for dependency}
\author{Li Ma \and Jialiang Mao}

\begin{document}
\maketitle

\begin{abstract}
\doublespacing
We introduce a method---called Fisher exact scanning (FES)---for testing and identifying variable dependency that generalizes Fisher's exact test on $2\times 2$ contingency tables to $R\times C$ contingency tables and continuous sample spaces. FES proceeds through scanning over the sample space using windows in the form of $2\times 2$ tables of various sizes, and on each window completing a Fisher's exact test. Based on a factorization of Fisher's multivariate hypergeometric (MHG) likelihood into the product of the univariate hypergeometric likelihoods, we show that there exists a coarse-to-fine, sequential generative representation for the MHG model in the form of a Bayesian network, which in turn implies the mutual independence (up to deviation due to discreteness) among the Fisher's exact tests completed under FES. This allows an exact characterization of the joint null distribution of the $p$-values and gives rise to an effective inference recipe through simple multiple testing procedures such as \v{S}id\'{a}k and Bonferroni corrections, eliminating the need for resampling. In addition, FES can characterize dependency through reporting significant windows after multiple testing control. The computational complexity of FES is approximately linear in the sample size, which along with the avoidance of resampling makes it ideal for analyzing massive data sets. We use extensive numerical studies to illustrate the work of FES and compare it to several state-of-the-art methods for testing dependency in both statistical and computational performance. Finally, we apply FES to analyzing a microbiome data set and further investigate its relationship with other popular dependency metrics in that context.  
\end{abstract}

\doublespacing

\section{Introduction}

Testing independence and identifying dependency of two random variables has numerous applications. For example, pairwise dependence is often used as the basis for building various networks such as gene expression networks and social networks. Complex, nonlinear dependence structures are commonplace in such applications, which call for flexible, nonparametric methods for testing and characterizing them. 
This problem has drawn great attention from both the statistical and the computational communities, with methods proposed from several perspectives, including those from an information theoretic perspective through nonparametric estimates of mutual information \citep{kraskov:2004,kinney:2014} and the more recently introduced maximal information coefficient \citep{reshef:2011}; the distance correlation approach \citep{szekely:2007,szekely:2009} and the recently proposed $G^2$ statistic \citep{wang:2016} that generalize the classical notion of correlation; a Bayesian modeling approach that compares the goodness-of-fit of nonparametric models for independence versus that for dependence \citep{filippi:2015}. These are just a few examples among many others. Most of the existing methods focus on constructing a proper score for measuring the extent of dependency, while resorting to resampling methods such as bootstrapping and permutation to evaluate the statistical significance of the resulting score.

In the midst of the ``big data era'', nowadays data sets that require dependency analysis are often massive---involving up to millions or even billions of observations as well as many variable pairs for which dependency is of interest. This amount of data presents both opportunities and challenges. First, with such large data sets, it becomes possible to identify dependencies that are otherwise impossible to detect. In particular, very {\em weak} dependencies buried in high noise backgrounds and {\em local} dependencies involving but a small subset of the observations now become potential inferential targets. At the same time, to analyze such big data in a flexible manner, one must construct methods that are computationally efficient both in CPU time and in memory requirement, while maintaining statistical soundness. 

We introduce a framework that satisfies these needs. 
It is called Fisher exact scanning (FES) because it marries two simple and powerful inference techniques---(i) Fisher's exact test for testing dependency on $2\times 2$ tables conditional on the marginals \citep{fisher:1954} and (ii) multi-scale scanning \citep{rufibach:2010,walther:2010}. Under this framework one scans over the sample space using windows that are $2\times 2$ contingency tables of various sizes, and on each table completes a Fisher's exact test. The key to effective inference under FES lies in proper multiple testing control in combining the results from these tests.

FES inherits several desirable features from Fisher's exact test and from multi-scale scanning. First, the null sampling distribution of the test statistics, or equivalently those of the $p$-values, is available analytically in an exact manner. As we will show, this is because inference under FES is conditional on the marginal observations, and that due to a factorization of Fisher's multivariate hypergeometric (MHG) likelihood on general $R \times C$ contingency tables into the product of (univariate) hypergeometric (HG) likelihoods on $2\times 2$ subtables, under the null hypothesis of independence the $p$-values from the scanning windows are mutually independent (up to deviations caused by discreteness). Effective multiple testing adjustment can thus proceed based on the exact null behavior of the $p$-values through simple techniques such as \v{S}id\'{a}k or Bonferroni correction and common false discovery rate controls without resorting to resampling. This makes FES particularly appealing in ``big data'' settings where each application of a dependency test can be computationally demanding and thus resampling can become extremely computationally intense. 

Fisher was a proponent for such conditional inference as he argued that the marginals are ``almost ancillary'' for the dependency. Whether conditioning on the marginals is desirable has been a point of controversy, which we do not attempt to settle. See \cite{little:1989} for a historical review of the issue and reasons why such conditioning is desirable, and see \cite{choi:2015} for a recent exploration on this. FES maintains this feature and generalizes the conditional dependency test from $2\times 2$ tables to $R \times C$ tables and continuous variables.

As a variant of multi-scale scanning, FES attains (i) computational efficiency---with an amount of computation scaling approximately linearly with the sample size and fixed memory requirement for any given maximum resolution of scanning; and (ii) the ability to not only test the existence of dependence, but to identify the actual nature of the dependence---through reporting the $2\times 2$ subtables on which the $p$-values are deemed significant after multiplicity adjustment. 

It is worth noting that while the current work mainly focuses on testing variable dependency, according to our knowledge, the likelihood factorization of the MHG model into HG likelihoods has not been reported previously and is of its own value beyond the scope of this work. For example, it gives a generative Monte Carlo strategy for the MHG, and likely has many more applications. Readers familiar with multi-scale modeling may recall other likelihood factorizations involving decomposing a whole likelihood into the product of likelihoods defined on nested windows (see for example \cite{kolaczyk:2004,ma:2016}), which is exploited in \cite{filippi:2015}. The MHG factorization reported here is fundamentally different as the $2\times 2$ tables on which the HG likelihoods arise can be partially overlapping and non-nested, making the factorization particularly interesting. We will show that conditioning on the marginals is critical here. Another classical example of orthogonal decomposition of information in the context of contingency tables is the sequential partitioning of $\chi^2$-statistics into multiple $\chi^2$-statistics defined on subtables \citep{lancaster:1949}. This classical decomposition also requires the sequence of tables to be nested and cannot be partially overlapping. See Sec.~3.3.3 in \cite{agresti2013categorical} for more details.

The rest of the paper is organized as follows. In \ref{sec:method}, we start by introducing a multi-scale discrete characterization of variable dependency based on coarse-to-fine partitioning on the margins. Then we establish the likelihood factorization for the MHG distribution into a product of HG likelihoods and present the induced sequential generative Bayesian network representation of the MHG. Finally we introduce the FES method and derive inference recipe using the likelihood factorization. In \ref{sec:num_exam}, we carry out comprehensive numerical studies to evaluate both the statistical performance and the computational efficiency of FES, and compare it to those of a number of state-of-the-art and classical methods. 
In \ref{sec:data_exam} we illustrate the work of FES in analyzing a publicly available microbiome data set, and in particular show how one may use FES as a tool for evaluating the statistical significance for other dependency metrics. We conclude in \ref{sec:discussion} with brief remarks. 
All proofs are given in the Supplementary Material. An {\tt R} package for FES is available freely on {\tt Github}.

We close the introduction by connecting our work to some particularly relevant references. \cite{gretton:2007} and \cite{heller:2016} also recognize the large computational demand for carrying out resampling to evaluate statistical significance. \cite{gretton:2007} proposes an independence test based on the Hilbert-Schmidt norm of the cross-covariance operator given a chosen kernel, and quantifies the sampling behavior of the resulting dependency metric HSIC under the null hypothesis of independence using asymptotic Gamma approximation. Like the FES, \cite{heller:2016} also proposes a strategy for statistical significance under an exact finite sample null distribution. Moreover \cite{heller:2016} is also based on a partitioning of the sample space into contingency tables. Instead of a divide-and-conquer scanning strategy, however, \cite{heller:2016} appeals to $\chi^2$ and likelihood ratio type statistics for testing a whole partition of the sample together. 
Some other works that also aim at testing independence based on a partition of the sample space include \cite{heller:2012} and \cite{zhang:2017}. \cite{heller:2012} proposes a strategy to reorganize a data set into a $2\times 2$ contingency table and then adopts a $\chi^2$-like statistic to measure dependency on the resulting $2\times 2$ table. The strategy for the reorganization of the table is through thresholding on the distance from a baseline observation and sum over all possible baselines, and permutation is needed for judging significance. In comparison to the partitioning strategy proposed in our work, using pairwise distances to divide the original data set as in \cite{heller:2012} enjoys the advantage of easy application to multivariate and high-dimensional situations. A very recent method BET \citep{zhang:2017} also appeals to a sequence of partitioning on the sample space to form a cascade of contingency tables (i.e., the strata). Instead of taking a local scanning approach as in FES, BET carries out a single test for each whole stratum. Both FES and BET fall into the general category of multi-scale methods. One way to understand the difference is through an analogy to multi-resolution methods for nonparametric inference---FES is based on a location-scale decomposition of the dependency structure, analogous to time-frequency decomposition of functions under wavelet bases, whereas BET uses a scale-only decomposition, analogous to the frequency decomposition under Fourier transforms.

\section{Method}
\label{sec:method}

\subsection{A multi-scale characterization of variable dependency}

For a pair of jointly distributed random variables $(X,Y)$, we are interested in testing whether $X\independent Y$, and in case $X\not \independent Y$, characterizing the nature of their dependence. In the following, for simplicity, we shall assume that the data $(X_1,Y_1), (X_2,Y_2),\ldots, (X_n,Y_n)$ are $n$ i.i.d., draws from the joint distribution of $(X,Y)$, denoted by $F$.  It is worth noting that the FES framework we present can be applied in exactly the same way without modification  when (i) the $X_i$ values are fixed and the $Y_i$'s are conditionally independent draws given $X$ from the conditional distribution of $Y$ given $X$, $F_{Y|X}$, or (ii) the $\{(X_i,Y_i):i=1,2,\ldots,n\}$ jointly is a draw from $F^{n}$ conditional on the marginal order statistics $\{X_{(1)},X_{(2)},\ldots,X_{(n)},Y_{(1)},Y_{(2)},\ldots,Y_{(n)}\}$ because as we will see inference under FES is conditional on the marginal order statistics.

To begin, let us consider the most simple situation in which both variables are binary, that is, $(X_i,Y_i)\in\{0,1\}\times \{0,1\}$. In this case, the data can be summarized on a $2\times 2$ contingency table. 
Specifically, we let $\bn=(n_{00},n_{01},n_{10},n_{11})$ be the counts in that table with $n_{ab}=|\{i:X_i=a,Y_i=b\}|$ being the number of $(X_i,Y_i)$ pairs with $X_i=a$ and $Y_i=b$  for $a,b\in \{0,1\}$. Similarly, let $\pi_{ab}=F(X_i=a,Y_i=b)$ be the cell probability for each cell in the table. (We use $|\cdot|$ to denote the cardinality of a set.) A natural quantity that measures the extent of dependency is the so-called {\em odds-ratio} (OR). That is,
\[
{\rm OR}= \frac{\pi_{00}/\pi_{01}}{\pi_{10}/\pi_{11}}
\]
It is often convenient to work with the log odd-ratios (LOR), denoted as $\theta=\log {\rm OR}$. 

Conditional on the row and column marginal totals, the distribution of such a $2 \times 2$ table 
is a one-parameter exponential family with $\theta$ being the natural parameter and $n_{11}$ the corresponding sufficient statistic. Under this model, a test for independence between the row and the column variables is the two-sided Fisher's exact test for
\vspace{-0.8em}
\[
H_0: \theta = 0,
\]
\vspace{-2.8em}

\noindent which rejects when $\hat{\theta}<C_1$ or $\hat{\theta}>C_2$, where $\hat{\theta}=\log \frac{n_{00}/n_{01}}{n_{10}/n_{11}}$ is the empirical LOR. The $p$-value is available {\em exactly} based on the tail probabilities of the (central) HG distribution. 

When $X$ and $Y$ are not necessarily binary, let $\om=\om_{X}\times \om_{Y}$ be the joint sample space of $F$. For illustration (in fact without loss of generality) let us assume that $\om=[0,1]\times [0,1]$. We take a coarse-to-fine multi-scale discretization approach to characterizing the relationship between $X$ and $Y$. At the coarsest level, we view $\om$ as the union of  four disjoint pieces $\om_{00}=[0,1/2)\times [0,1/2)$, $\om_{01}=[0,1/2)\times [1/2,1]$, $\om_{10}=[1/2,1]\times [0,1/2)$, and $\om_{11}=[1/2,1]\times [1/2,1]$ such that $\om=\om_{00}\cup\om_{01}\cup\om_{10}\cup\om_{11}$. Under this discretization, the data can be viewed as a $2\times 2$ table $(n(\om_{00}),n(\om_{01}),n(\om_{10}),n(\om_{11}))$ 
where $n(\om_{ab})=|\{i:(X_i,Y_i)\in\om_{ab}\}|$ for $a,b\in\{0,1\}$.

Similarly, for any $A=[l_x,u_x]\times [l_y,u_y]\subset \om$, we define an associated $2\times 2$ table by dividing $A$ into four disjoint pieces $A_{00}$, $A_{01}$, $A_{10}$, and $A_{11}$ such that 
\[ A_{ab} = [l_{x}+a/2\cdot (u_x-l_x),l_{x}+(a+1)/2\cdot (u_x-l_x)) \times [l_{y}+b/2\cdot (u_{y}-l_{y}),l_{y}+(b+1)/2\cdot (u_y-l_y)] \]
for $a,b\in\{0,1\}$. (Note that it does not matter which ends of the intervals are open/closed insofar as the $A_{ab}$'s form a partition of $A$.) In the following, we shall let $A_{a\cdot} = A_{a0}\cup A_{a1}$ and $A_{\cdot b} = A_{0b}\cup A_{1b}$. Also, we let $n(A)=|\{i:(X_i,Y_i) \in A\}|$, and for each $A$ we let
\[
\theta(A) = \log \frac{F(A_{00})F(A_{11})}{F(A_{10})F(A_{01})} \quad \text{and} \quad \hat{\theta}(A) =\log \frac{n(A_{00})n(A_{11})}{n(A_{01})n(A_{10})} 
\]
be the associated LOR and empirical LOR for the $2\times 2$ subtable on $A$.

Next we introduce a few more notions regarding discrete approximation to continuous sample spaces, which will serve as a building block for the FES framework.

\begin{defn}[Level-$(k_1,k_2)$ windows and stratum]
For any $k_1,k_2=0,1,2,\ldots$, we call a set $A$ a level-$(k_1,k_2)$ {\em window} if it is of the form
\[
A = I^{k_1}_{l_1} \times I^{k_2}_{l_2} 
\]
for some $l_1 \in \{1,2,\ldots,2^{k_1}\}$ and $l_2 \in \{1,2,\ldots,2^{k_2}\}$, where for any $k\geq 0$ and $l=1,2,\ldots,2^k$, $I^{k}_{l}=[(l-1)/2^{k},l/2^{k})$. We call the collection of all level-$(k_1,k_2)$ windows the {\em $(k_1,k_2)$-stratum}, and denote it as $\A^{k_1,k_2}$. That is,
\[
\A^{k_1,k_2} = \mathcal{I}^{k_1}\times \mathcal{I}^{k_2} \quad \text{with $\mathcal{I}^{k} = \{I^{k}_{l}:l=1,2,\ldots,2^{k}\}$ for each $k$.}
\]

\end{defn}
\noindent Remark: Intuitively, the $(k_1,k_2)$-stratum is a discretization of the sample space using  $2^{k_1}$ and $2^{k_2}$ categories for the two margins respectively. 

\begin{defn}[$(k_1,k_2)$-independence] 
\label{defn:k1k2_indep}
For any $k_1,k_2=0,1,2,\ldots$, we say that $X$ and $Y$ are {\em $(k_1,k_2)$-independent}, and write it as $X\independentk Y$, if for any $l_1\in \{1,2,\ldots,2^{k_1}\}$ and $l_2\in \{1,2,\ldots,2^{k_2}\}$,
\[
F(I^{k_1}_{l_1}\times I^{k_2}_{l_2}) = F_{X}(I^{k_1}_{l_1}) F_{Y}(I^{k_2}_{l_2}),
\]
where $F_{X}$ and $F_{Y}$ are the corresponding marginal distributions of $X$ and $Y$.
\end{defn}
The meaning of $(k_1,k_2)$-independence is that if we approximate the sample space of $F$ using the $(k_1,k_2)$-stratum, then the discretized versions of $X$ and $Y$, which are jointly distributed multinomials, are independent. The next lemma gives an equivalent characterization of $(k_1,k_2)$-independence that is often easier to check than the original definition.
\begin{lem}
\label{lem:equiv_k_indep}
$X$ and $Y$ are $(k_1,k_2)$-independent if and only if for any $l_1,l_1' \in \{1,2,\ldots,2^{k_1}\}$ and $l_2,l_2'\in  \{1,2,\ldots,2^{k_2}\}$,
we have
\[
F(I^{k_1}_{l_1}\times I^{k_2}_{l_2})F(I^{k_1}_{l_1'}\times I^{k_2}_{l_2'}) = F(I^{k_1}_{l_1}\times I^{k_2}_{l_2'})F(I^{k_1}_{l_1'}\times I^{k_2}_{l_2}).
\]
\end{lem}

The next lemma states that if $X$ and $Y$ are independent on the $(k_1,k_2)$-stratum, then they are independent at all coarser strata. We say that stratum-$(k_1',k_2')$ is {\em coarser} than stratum-$(k_1,k_2)$ if $k_1'\leq k_1$ and $k_2'\leq k_2$, (and {\em finer} if $k_1'\geq k_1$ and $k_2'\geq k_2$). 
\begin{lem}
\label{lem:k_independence}
If $X$ and $Y$ are $(k_1,k_2)$-independent, then they are $(k_1',k_2')$-independent for all $0\leq k_1'\leq k_1$ and $0\leq k_2' \leq k_2$.
\end{lem}

It is not hard to see that $X\independent Y$ implies $X\independentk Y$ for all $k_1$ and $k_2$. The next theorem states that the reverse is also true.
\begin{thm}
\label{thm:independent_k}
\[
X\independent Y \quad \Leftrightarrow \quad X\independentk Y \quad \text{for all $k_1,k_2=0,1,2,\ldots$}
\]
\end{thm}

The relationship between $X\independent Y$ and $X\independentk Y$ suggests a natural multi-scale strategy for nonparametrically testing independence---simply through testing $(k_1,k_2)$-independence from coarse to fine strata. In practice, one is typically only interested in the dependency up to a practical level of details. In other words, $(k_1,k_2)$-independence for some large enough $k_1$ and $k_2$ is in fact what is sought after by practitioners in application areas.

The next question is, then, how to effectively test for $(k_1,k_2)$-independence. 
A brute-force strategy based on classical tests applied on the entire $(k_1,k_2)$-stratum, such as a $\chi^2$-test, will face two fundamental difficulties. First, for even just moderately large $k_1$ and $k_2$, the tests would incur very many degrees of freedom, $(2^{k_1}-1)(2^{k_2}-1)$ to be exact. On the other hand, at such a fine discretization, most, if not all, of the $2^{k_1+k_2}$ windows will typically contain only a small number of observations thereby making the asymptotic approximation to the sampling distributions unreliable.

To overcome these difficulties, we seek an alternative strategy that aims to be prudent in ``using up the degrees of freedom'' in the test. 
The following theorem gives an alternative way to characterizing $(k_1,k_2)$-independence in terms of ORs on $2\times 2$ subtables in coarser stratifications, which will serve as the basis for a coarse-to-fine scanning test strategy. From now on, we use $\A^{(k_1,k_2)}=\cup_{k_1'\leq k_1, k_2'\leq k_2}\A^{k_1',k_2'}$ to denote the totality of all windows in coarser strata than $\A^{k_1,k_2}$.
\begin{thm}
\label{thm:independent_k_or}
For any $k_1,k_2=1,2,\ldots$, 
\[
X\independentk Y \quad \Leftrightarrow \quad \theta(A)=0 \quad \text{for all $A\in \A^{(k_1-1,k_2-1)}$}.
\]
\end{thm}
Note that there are $(2^{k_1}-1)(2^{k_2}-1)$ elements in $\A^{(k_1-1,k_2-1)}$. 
Hence the cell probabilities of the $(k_1,k_2)$-stratum when all $\theta(A)=0$ has $2^{k_1+k_2} - 1 - (2^{k_1}-1)(2^{k_2}-1)=2^{k_1}+ 2^{k_2} - 2$ free parameters, which matches the degrees of freedom in that table under $X\independentk Y$. This theorem will allow us to transform a complex alternative of $(2^{k_1}-1)(2^{k_2}-1)$ degrees of freedom into $(2^{k_1}-1)(2^{k_2}-1)$ simple alternatives each of 1 degree of freedom in testing $(k_1,k_2)$-independence.

\subsection{Likelihood factorization on contingency tables given margins}

We next establish the main theoretical result that will help us derive inference recipes for FES---namely a factorization of the likelihood under Fisher's multivariate hypergeometric (MHG) distribution into the product of a collection of (univariate) HG likelihoods defined on $2\times 2$ subtables corresponding to the windows in coarser strata. 

Though there is a more general version of the theorem for $R\times C$ contingency tables with $R$ and $C$ greater than 1, we shall describe it in the particular case when $R=2^{k_1}$ and $C=2^{k_2}$ as it is the current context. (The proof of the theorem applies to the more general case with only minor modifications.) Hereafter, for $i,j \geq 0$, we use $\bn_{i,j}=\{n(A):A\in \A^{i,j}\}$ to represent the $2^{i}\times 2^{j}$ contingency table corresponding to the $(i,j)$-stratum.

\begin{thm}[Multi-scale factorization of the multivariate hypergeometric likelihood]
\label{thm:likelihood_factorization}
Suppose the counts $\bn_{k_1,k_2}=\{n(A):A\in\A^{k_1,k_2}\}$ in a $2^{k_1}\times 2^{k_2}$ contingency table arise from Poisson, multinomial, or product-multinomial sampling. Then if the two marginal variables are independent, the conditional sampling probability given the row totals $\bn_{k_1,0}=\{n(A):A\in \A^{k_1,0}\}$ and column totals $\bn_{0,k_2}=\{n(A):A\in\A^{0,k_2}\}$, which is a (Fisher's) MHG likelihood, factorizes into the product of the likelihood of the HG likelihood on the $2\times 2$ subtables on all $A\in \A^{(k_1-1,k_2-1)}$. That is,
\[
p(\bn_{k_1,k_2}\,|\,\bn_{k_1,0},\bn_{0,k_2}) = 
\prod_{A\in \A^{(k_1-1,k_2-1)}} p(n(A_{00})\,|\,n(A_{0\cdot}),n(A_{\cdot 0}),n(A))
\]
where $p(\bn_{k_1,k_2}\,|\,\bn_{k_1,0},\bn_{0,k_2})$ is the MHG pmf for the whole table given the marginal totals, and $p(n(A_{00})\,|\,n(A_{0\cdot}),n(A_{\cdot 0}),n(A))$ the HG pmf on the $2\times 2$ subtable on each $A$ given its row and column totals.
\end{thm}

\noindent Remark I: In fact, the proof of the theorem (see Supplementary Materials~S1) implies that
\[
p(\bn_{k_1,k_2}\,|\,\bn_{k_1',0},\bn_{0,k_2'}) = p(\bn_{k_1,k_2}\,|\,\bn_{k_1,0},\bn_{0,k_2}) \quad \text{for any $k_1'\geq k_1$ and $k_2' \geq k_2$.}
\]
By letting $k_1',k_2'\uparrow \infty$, we see that the same factorization of the probability of $\bn_{k_1,k_2}$ holds even if we condition on the marginal order statistics for $X$ and $Y$. Intuitively, under independence, additional knowledge about the marginals does not inform us about the conditional distribution on the $(k_1,k_2)$-stratum beyond what the corresponding discretization does. 
\vspace{0.5em}

\noindent 
Remark II: It is worth noting that the theorem regards the conditional distribution of $\bn_{k_1,k_2}$ given the marginal order statistics, and it remains valid regardless of the marginal distributions of $X$ and $Y$. An implication is that after any monotone transformations to the two margins $X$ and $Y$ (or equivalently varying the strata partitioning points along either margin) will not affect the validity of the factorization. Of course, after such transforms the resulting contingency tables will be different, but the theorem holds still.
\vspace{0.5em}

The likelihood factorization implies a form of orthogonality in an information theoretic sense---the empirical evidence contained in each of the $2\times 2$ tables is non-overlapping once we condition on the corresponding marginal totals. At first glance, this result is surprising because the $2\times 2$ tables can either be nested or partially overlapping (so non-nested) with each other as illustrated in \ref{fig:overlap_windows}. The windows share observations and thus empirical evidence. The theorem suggests that the shared information is all contained in the marginals. 

\begin{figure}[t!]
\begin{center}
\includegraphics[width=16em,clip=TRUE, trim=0 20mm 0 20mm]{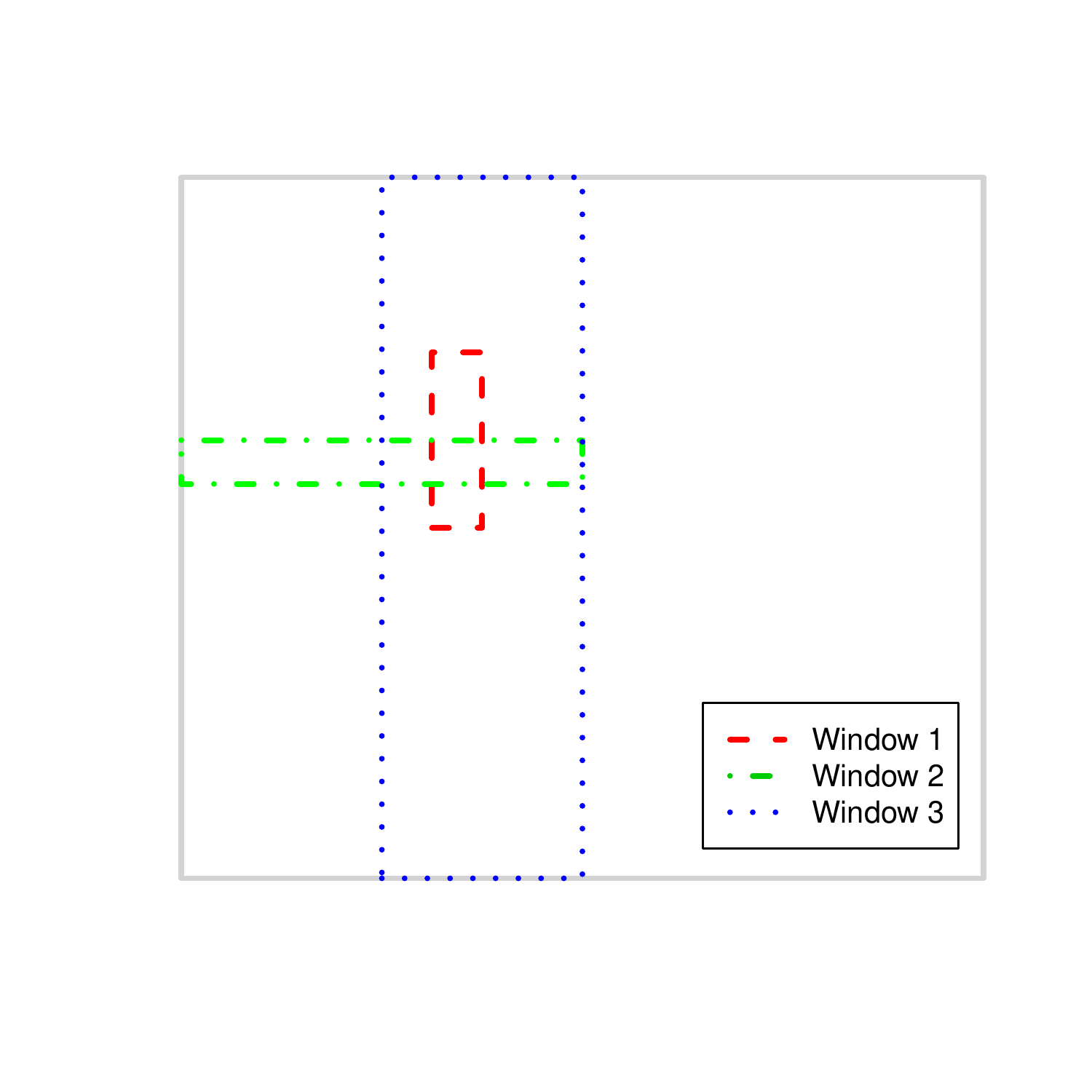}
\caption{Three example scanning windows. Window 1 is nested in Window 3, while Window 2 is partially overlapping with the other two.}
\label{fig:overlap_windows}
\end{center}
\vspace{-1em}
\end{figure}

\begin{figure}[t!]
\begin{center}
\includegraphics[width=27em, clip=TRUE, trim=90mm 0mm 50mm 0mm]{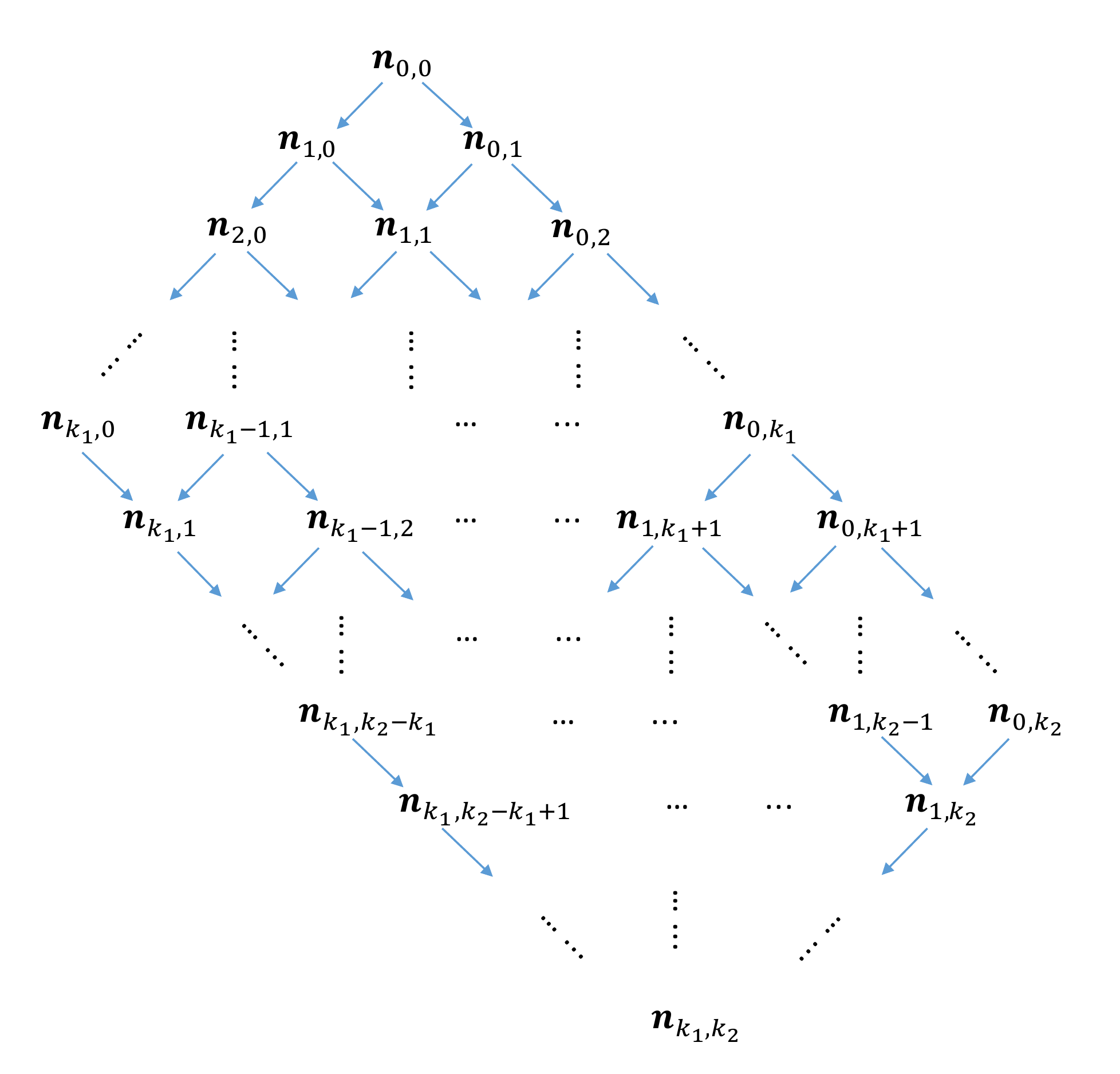}
\caption{A graphical representation in terms of a Bayesian network of the sequential generative model for the MHG distribution implied by Theorem~\ref{thm:likelihood_factorization}. Without loss of generality, here we assume $k_1\leq k_2$. We use $\bn_{i,j}=\{n(A):A\in \A^{i,j}\}$ to represent the $2^{i}\times 2^{j}$ table corresponding to the $(i,j)$-stratum. Given the marginals---which are those tables on the topleft and topright sides---every inner stratum is generated independently from HG distributions on the corresponding $2\times 2$ subtables given their marginals determined in their parent strata.}
\label{fig:filtration_tree}
\end{center}
\vspace{-1em}
\end{figure}

Theorem~\ref{thm:likelihood_factorization} provides a multi-scale sequential generative model for the MHG distribution, based on generating strata coarse-to-fine using independent draws from HG distributions given the coarser strata. The model can be represented as a directed acyclic graph (DAG), i.e., a Bayesian network, illustrated in \ref{fig:filtration_tree}. 
One immediate application of the generative model is a Monte Carlo sampler for MHG distributions by sequentially generating the corresponding $2\times 2$ tables---each from an HG given the previously generated row and column totals---from coarse to fine strata. In the following, we use the sequential generative model to devise an inference recipe for identifying variable dependence.

\subsection{Fisher exact scanning}
\label{sec:fes}
Now we are ready to introduce FES for testing dependency. Theorem~\ref{thm:independent_k_or} suggests a ``divide-and-conquer'' strategy for testing and characterizing $(k_1,k_2)$-dependency between $X$ and $Y$, one based on testing $H_0(A):\theta(A)=0$ for all $A \in \A^{(k_1-1,k_2-1)}$. 

The question is then, what test to employ for testing each $H_0(A)$. Some obvious candidates include common tests such as the likelihood ratio test and Pearson's $\chi^2$ test. The significance level of these tests can be evaluated using the $\chi^2$ distribution when the number of counts is large. Such asymptotic assumptions may be appropriate for large subtables, i.e., such $A\in \A^{i,j}$ when $i$ and $j$ are small, but will typically not hold for the vast majority of the subtables in high resolutions.

We adopt Fisher's approach by carrying out a conditional test on $H_0(A)$ given the corresponding marginal totals. This choice is not merely because we desire to condition out the marginals, but as we shall see will result in very simple inference recipe due to the likelihood factorization of MHG.
In particular, we carry out Fisher's exact test on the $2\times 2$ tabulation of each window $A$. We refer to this strategy as {\em Fisher exact scanning} (FES) as it is essentially multi-scale scanning using Fisher's exact test.  The most simple version of FES is  through carrying out Fisher's exact test exhaustively on all windows $A\in \A^{(k_1-1,k_2-1)}$. In practice, however, many, if not all, of the windows in the fine strata contain so few data that they cannot provide strong empirical evidence. As such, one can carry out Fisher's exact test on strata up to some maximum resolution. From now on, we call $i+j$ the {\em resolution} of $\A^{i,j}$, and so one may choose a maximum resolution $M\leq k_1+k_2-2$, and scan over windows in all $\A^{i,j}$ with $0\leq i+j\leq M$ using Fisher's exact test.

We summarize the entire FES procedure, including both an optional pre-processing step involving monotonically transforming the marginals and post-processing involving evaluating statistical significance and reporting significant findings in Algorithm~\ref{alg:FES}. We will explain the key components in the procedure in the rest of this subsection but readers may refer to Algorithm~\ref{alg:FES} for an overview of the method.

\begin{algorithm}[p]
\caption{Fisher exact multi-scale scanning (FES) with three-stage \v{S}id\'{a}k correction}
\label{alg:FES}

\begin{algorithmic}
\\
\State{Apply empirical CDF (or rank) transform to each margin:} \Comment{Optional preprocessing}
\[ (x_i,y_i)\rightarrow (\hat{F}_{X}(x_i),\hat{F}_{Y}(y_i)) \quad \text{for each observation $(x_i,y_i)$} \vspace{0.5em}\]

\For{$r$ in $0,1,2,\ldots,M \leq k_1+k_2-2$}  \Comment{Scan from low to high resolutions}\\
\vspace{-0.5em}

\For{$i$ in $0,1,2,\ldots,\min\{k_1-1,r\}$}
\State{$j = r-i$}
\State{$L(i,j) = 0$} \Comment{Initialize the test counter for level $(i,j)$}\\
\vspace{-0.5em}

\For{each $A\in \A^{i,j}$}\\
\vspace{-0.5em}

\If{$S(A)=1$} \Comment{If $A$ passes the screening rule}
\State{Compute the $p$-value $p(A)$ for Fisher's exact test on the $2\times 2$ table on $A$.}
\State{$L(i,j) \leftarrow L(i,j)+1$} \Comment{Update the test counter}
\Else
\State{Skip testing on $A$.} \Comment{When $A$ does not pass screening, simply skip it.}
\EndIf\\
\vspace{-0.5em}

\EndFor\\
\vspace{-0.5em}

\State{Compute \v{S}id\'{a}k's $p$-value for the tests on $\A^{i,j}$:} \Comment{Multiplicity control per stratum}
\[
p(i,j) = 1 - \left(1 - \min_{A\in\A^{i,j}}p(A)\right)^{L(i,j)}.
\]
\EndFor\\
\vspace{-0.5em}

\State{Compute \v{S}id\'{a}k's $p$-value for resolution $r$:} \Comment{Multiplicity control per resolution}
\[
p_{resol}(r) = 1 - \left(1 - \min_{(i,j):i+j = r, L(i,j)>0} p(i,j)\right)^{|\{(i,j):i+j = r, L(i,j)>0\}|}.
\]
\EndFor\\
\vspace{-0.5em}

\State{Compute the overall \v{S}id\'{a}k's $p$-value:} \Comment{Overall multiplicity control}
\[
p_{overall} = 1 - \left(1 - \min_{r} p_{resol}(r)\right)^{M+1}.
\]
\State{Reject the null hypothesis of independence at level $\alpha$ if $p_{overall} < \alpha$.}
\\
\State{Report windows with $p(A)<\alpha(A)$  where} \Comment{Report significant windows}
\[ \alpha(A) = 1-(1-\alpha)^{1/(M+1)\cdot 1/T(r)\cdot 1/L(i,j)} \text{ for all $A\in \A^{i,j}$.} \]
\end{algorithmic}
\end{algorithm}

\vspace{0.5em}

{\em Multiplicity adjustment.} FES transforms the characterization of arbitrary dependency structure into a multiple testing problem---through testing a collection of hypotheses on $2\times 2$ subtables of various sizes. In summarizing the results from the subtables, one must properly adjust for multiple testing.
We will show in the following that the likelihood factorization of the MHG justifies extremely convenient means to multiplicity adjustment. 

In particular, as we shall see, when $X$ and $Y$ are independent, the Fisher's exact tests, e.g., in terms of their $p$-values, are {\em mutually independent} (up to deviations from independence due to the discreteness of the HG distributions).
Again this may first appear puzzling because the windows can overlap (either nested or partly so). Thus one might not have expected the test statistics to be independent of each other. We show next that it follows from Theorem~\ref{thm:likelihood_factorization} through a data augmentation for the sequential generative model for MHG (\ref{fig:filtration_tree}).

 Specifically, suppose on each $A\in \A^{(k_1-1,k_2-1)}$, for all possible triplets of integers $(a,b,c)$ satisfying $0\leq a \leq c$, $0\leq b \leq c$, and $c\geq 1$, we generate a collection of mutually independent random variables 
$\{n_{a,b,c}(A): 0\leq a \leq c, 0\leq b \leq c, c\geq 1,A\in\A^{(k_1-1,k_2-1)}\}$ such that $n_{a,b,c}$ has the HG distribution with first row total $a$, first column total $b$, and overall total $c$, which from now on we shall denote as ${\rm HG}_{a,b,c}$. Now starting from the coarsest stratum, $\A^{0,0}$, we generate the strata coarse-to-fine from the sequential generative mechanism. Suppose all coarser strata have been generated, then for each $A\in\A^{i,j}$, we draw the corresponding $2\times 2$ table by letting $n(A_{00}) = n_{n(A_{0\cdot}),n(A_{\cdot 0}),n(A)}(A)$, that is, we let $n(A_{00}) = n_{a,b,c}(A)$ for $a=n(A_{0\cdot})$, $b=n(A_{\cdot 0})$, and $c=n(A)$. Here the coarser $2\times 2$ tables serve as the {\em selector} variables that determine which random variables are observed in the finer strata, but do not affect the sampling distributions of the latter otherwise. 

Now on each $A$ and for each $(a,b,c)$ combination, we can define $p_{a,b,c}(A)$ to be the corresponding two-sided Fisher's exact test $p$-value for $n_{a,b,c}(A)$ under ${\rm HG}_{a,b,c}$. Then, under $H_0(A)$, the realized $p$-value for the test on each $A$ is given by
\[
p(A) = p_{n_{0\cdot}(A),n_{\cdot 0}(A),n(A)}(A)
\]
and thus ${\rm P}(p(A)\leq \alpha\,|\,H_0(A))=_{dsc} \alpha$, where and hereafter ``$=_{dsc}$'' means ``equal up to deviations caused by discreteness''. For example, here ${\rm P}(p(A)\leq \alpha\,|\,H_0(A))$ takes the largest attainable value not exceeding $\alpha$.

Next, we argue that for all $r\geq 0$, and $\alpha_{A}\in [0,1]$ for all $A$,
\[ {\rm P}(p(A)\leq \alpha_{A}\,\, \forall A\in \cup_{i+j\leq r}\A^{i,j} \,|\,X\independent Y)=_{dsc} \prod_{A\in \cup_{i+j\leq r}\A^{i,j}} \alpha_{A}.\]
Intuitively, this means that $p(A)$'s are as mutually independent as possible as allowed under the discreteness of the sample space. This is true because
\begin{align*}
&{\rm P}(p(A)\leq \alpha_{A}\,\,\forall A\in \A^{i,j} \text{ with }i+j \leq r\,|\,X\independent Y)\\
=&{\rm E}\big\{
\I\left(p(A)\leq \alpha_A \,\,\forall A\in \A^{i,j}\text{ with }i+j\leq r-1\right)\\
& \times {\rm P}\left(p(A)\leq \alpha_A\,\,\forall A\in\A^{i,j}\text{ with }i+j=r\,|\,\text{all $\bn_{i,j}$ such that $i+j=r+1$},X\independent Y\right)\,|\,X\independent Y\big\}
\end{align*}
where $\I(\cdot)$ is the indicator function for an event. But since
\begin{align*}
&\,\,\,\,\,\,{\rm P}\left(p(A)\leq \alpha_A\,\,\forall A\in\A^{i,j}\text{ with }i+j=r\,|\,\text{all $\bn_{i,j}$ such that $i+j=r+1$},X\independent Y\right)\\
&={\rm P}\left(p_{n_{0\cdot}(A),n_{\cdot 0}(A),n(A)}(A)\leq \alpha_A\,\,\forall A\in\A^{i,j}\text{ with }i+j=r\,|\,\text{all $\bn_{i,j}$ such that $i+j=r+1$},X\independent Y\right)\\
&=_{dsc} \prod_{A\in\A^{i,j}:i+j=r} \alpha_{A}, 
\end{align*}
by iteratively applying this argument for $i+j= r,r-1,\ldots,0$, we have
\begin{align*}
&{\rm P}(p(A)\leq \alpha_{A}\,\,\forall A\in \A^{i,j} \text{ with }i+j \leq r\,|\,X\independent Y)\\
=_{dsc} &{\rm P}(p(A)\leq \alpha_{A}\,\,\forall A\in \A^{i,j} \text{ with }i+j \leq r-1\,|\,X\independent Y)\times \prod_{A\in\A^{i,j}:i+j=r} \alpha_{A}\\
=_{dsc} & \prod_{A\in\A^{i,j}:i+j\leq r} \alpha_{A}.
\end{align*}

Due to the independence (modulo discreteness) among the $p$-values, simple strategies for controlling the family-wise error rate (FWER) such as \v{S}id\'{a}k and Bonferroni correction to $p$-values are effective. There are different ways to applying such corrections, our preferred strategy is to correct the $p$-value in three stages: (i) first among the windows in each stratum $\A^{i,j}$, then (ii) among all strata in each resolution $r$---i.e., for those $\A^{i,j}$ with $i+j=r$---for $r=0,1,\ldots,M$, and (iii) across the $M+1$ resolution levels. Adjusting multiplicity in these stages takes into account the fact that there are many more windows in finer strata and thus treating all tests equally across strata will result in overly large penalty on larger windows. 

Specifically,  in the first stage we aim to compute $p(i,j)$, the corrected $p$-value for the minimum $p$-value in each stratum $\A^{i,j}$. Let $L(i,j)$ be the number of windows in $\A^{i,j}$ on which Fisher's exact test is applied. Then with, for example, \v{S}id\'{a}k's correction,
\[
p(i,j) = 1 - \left(1 - \min_{A\in\A^{i,j}}p(A)\right)^{L(i,j)}.
\]
In the second stage, assume that we have carried out Fisher's exact scanning up to a maximum resolution level $M \leq k_1+k_2-2$. Then for each $r=1,2,\ldots,M$, let $p_{resol}(r)$ denote the corrected $p$-value for the $r$th resolution level, i.e., for all $\A^{i,j}$ with $i+j=r$. Let $T(r)$ be the number of $(i,j)$ pairs with $i+j=r$. (Here $T(r)=r+1$ but later when we introduce screening $T(r)$ is not necessarily $r+1$.) Again, with \v{S}id\'{a}k's correction,
\[
p_{resol}(r) = 1 - \left(1 - \min_{(i,j):i+j = r, L(i,j)>0} p(i,j)\right)^{T(r)}.
\]
Finally, let $p_{overall}$ be the ``overall'' corrected $p$-value. With \v{S}id\'{a}k's correction, it is
\[
p_{overall} = 1 - \left(1 - \min_{r} p_{resol}(r)\right)^{M+1}.
\]
The following theorem shows that using the overall $p$-value for rejecting/accepting the null hypothesis of independence achieves the desired level of FWER. 

\begin{thm}[FWER control]
\label{thm:fwer_control}
If the overall null hypothesis $H_0: X\independent Y$ is true,  
\[
{\rm P}(p_{overall} \leq \alpha\,|\,\bn_{k_1,0},\bn_{0,k_2}) \leq \alpha \quad \forall \alpha \in [0,1].
\]
\end{thm}

We emphasize that modern state-of-the-art nonparametric tests of dependency mainly rely on resampling such as permutation to obtain the proper significance threshold. The ability to achieve effective significance evaluation without resampling makes FES computationally desirable especially for large data sets.

\vspace{0.5em}

{\em Screening rules.} Windows in finer strata often contain few data points and thus due to the discreteness of the HG distributions cannot produce $p$-values small enough to be significant. Those windows can be skipped thereby reducing the number of tests and the multiple testing penalty incurred on the other windows. To this end we can adopt a screening rule for deciding whether a window should be tested or skipped. 
One simple criterion for screening is based on sample size thresholding---one can skip testing a $2\times 2$ table $A$ if $n(A)\leq s$ for some minimal required sample size $s$, and/or if min$(n(A_{0\cdot}),n(A_{1\cdot}),n_{\cdot 0}(A),n_{\cdot 1}(A)) \leq s'$, i.e., one of the row margins or one of the column margins is less than some threshold $s'$, because such a table cannot render very statistically significant $p$-value. 

More generally, we let $S(A)$ denote such a screening rule for each $2\times 2$ subtable $A$ with $S(A)=1$ indicating that $A$ passes the screening and thus a test is to be carried out on $A$, and $S(A)=0$ otherwise. When employing the screening rule, the three-stage multiplicity correction stays the same except that now $L(i,j)$ and $T(r)$ become random variables: $L(i,j)=\sum_{A\in\A^{i,j}} S(A)$, i.e., the number of tests carried out in $\A^{i,j}$, and $T(r)=|(i,j):i+j=r,L(i,j)>0|$, i.e., the number of resolutions in which at least one Fisher's exact test is applied. 

One concern about screening regards its effect on the multiple testing control and in particular on the independence among the $p$-values. To this end, one can check that provided that the screening rule $S(A)$ for each $A\in\A^{i,j}$ is measurable w.r.t.\ the $\sigma$-algebra generated by $\bn_{i+1,j}$ and $\bn_{i,j+1}$---e.g., when $S(A)$ is a function of $(n_{0\cdot}(A),n_{1\cdot}(A),n_{\cdot 0}(A),n_{\cdot 1}(A))$, the $p$-values will still be independent (up to deviations caused by discreteness) because the screening only modifies the selector and will not affect the sampling distribution of $p$-values given the selectors. Moreover, Theorem~\ref{thm:fwer_control} still holds with the new definition of $L(i,j)$ and $T(r)$. (We consider the general case with screening in the proof of that theorem.)

An optional but desirable property for the screening rule is {\em inheritability}, that is, if a window $A$ does not pass the screening, then any window contained in $A$ also does not pass the screening. When the screening rule has this property, screening becomes optional stopping. One can carry out FES from coarse to fine resolutions, and terminate the procedure on portions of the sample space as soon as a window does not pass the screening.

\vspace{0.5em}

{\em Identifying significant windows.}
A unique feature of the FES approach is its capability for identifying subsets of the data set that accounts for the detected dependency, if any. This feature is a consequence of the multi-scale scanning lineage, and is achieved simply through reporting the windows whose $p$-values are significant after multiple testing control. This feature is particularly useful when the underlying dependency is local in nature, involving only a small portion of the sample space and/or observations. 

With the three-stage \v{S}id\'{a}k's correction, at an overall FWER level $\alpha$, the null hypothesis of independence will be rejected if and only if there exists at least one window $A$ in some stratum $\A^{i,j}$ such that $p(A)\leq \alpha(i,j)$, where
$\alpha(i,j) = 1-(1-\alpha)^{1/(M+1)\cdot 1/T(r)\cdot 1/L(i,j)}$ is the adjusted level for the windows in that stratum. Thus we can report all windows whose $p$-values less than the corresponding threshold as significant. This will be illustrated in the numerical examples.

\vspace{0.5em}

{\em Optional preprocessing through marginal empirical CDF transform.} As stated in the remarks after Theorem~\ref{thm:likelihood_factorization}, inference under FES regards only the conditional distribution given the marginal order statistics, therefore monotonically transforming the marginal observations does not affect the validity of any of the previous theorems. A useful (though optional) preprocessing step is to apply the empirical CDF transform to each margin which turns the marginal observations into values such as $0,1/n,2/n,\ldots,(n-1)/n$.

One could of course apply FES on the original data without any marginal transform. We do recommend applying an empirical CDF transform (i.e., a rank transform) to the two margins first, because very often the marginal distributions (which are assumed unknown) are far from uniform and thus even under the null hypothesis of independence, some windows can contain a lot of observations while others in the same stratum may contain very few, causing the power for identifying a deviation from independence on the scanning windows to vary substantially across the sample space. Applying an empirical CDF helps even out the number of observations over the windows, thereby evening out the power of detecting any given level of dependency over the windows. The theoretical justifications for FES remains valid with or without the transform.
\vspace{0.5em}

{\em Choosing $k_1$, $k_2$, and $M$.} In applying FES, one needs to specify the maximum level of partitioning $k_1$ and $k_2$ for the two margins, as well as the maximum resolution $M$ for the scanning. After the empirical CDF transformation on each margin, the marginal order statistics lie on an equi-spaced grid, and so an upperbound for the values of $k_1$ and $k_2$ to be considered is $\lceil \log_2 n \rceil$, as any higher levels of partitioning will not generate any window with more than a single observation. In practice, however, it is unnecessary to choose such a large value. Instead, one can choose $k_1$ and $k_2$ in conjunction with the screening rule adopted. For example, if the screening rule is such that no table is tested with any row or column having less than $s'$ observations (e.g., $s'=10$), then $k_1$ and $k_2$ can be set to $\lfloor \log_2 (n/s') \rfloor$. 

The choice on $M$ can follow from a similar sample size consideration. Under the null hypothesis of independence, the data are on average evenly spread over an equi-spaced grid after the empirical CDF transformation, and so if we adopt a screening rule that only tests windows with at least $s$ observations (e.g., $s=25$), then a reasonable choice for $M+1$ is either $\lfloor \log_2 (n/s) \rfloor$ or some value slightly larger. The cost of choosing an $M$ too large lies in the multiple testing adjustment. The more resolutions are tested, the more penalty is incurred. In Supplementary Materials~S2, we carry out a sensitivity analysis of the performance of FES with respect to the choice of $k_1$, $k_2$, and $M$.

\vspace{0.5em}

{\em Conservativeness due to discreteness.} 
While we have focused our discussions on the desirable features that FES inherits from Fisher's exact test and multi-scale scanning, one undesirable property is also passed down. In particular, a commonly criticized drawback of Fisher's exact test is its conservativeness due to the discreteness of the HG distributions. For the same reason, FES also tends to be conservative, especially with small sample sizes and in higher resolutions where the numbers of data points in many windows are small. To address this issue, continuity corrections to the Fisher's exact test designed to overcome its conservativeness, such as the so-called ``mid $p$-value'' \citep[p.17]{agresti2013categorical}, can  be adopted in FES. Our experience suggests that this can substantially attenuate, though often not completely eliminate, the conservativeness. (See \ref{fig:null_power} for example.) We apply the mid $p$-value correction in all of our numerical studies and implement it in our software. 

\vspace{0.5em}

{\em Large sample consistency.} Next we investigate the behavior of the three-stage \v{S}id\'{a}k correction as sample size increases. In particular, it would be reassuring if FES results in a consistent test as sample size grows with respect to any alternative. That is, the power of the test for detecting any arbitrary alternative from independence converges to 1 as the total sample size increases. There are two relevant types of consistency in the current context---(i) the consistency in rejecting the global null of independence $H_0:X\independent Y$, i.e., the probability of rejecting $H_0$ converges to 1 as $n\rightarrow \infty$ when $X\not\independent Y$, and (ii) the consistency of rejecting each local null $H_0(A):\theta(A)= 0$ for every $A$ such that $\theta(A)\neq 0$ at the significance level adjusted for multiple testing. The second type of consistency is stronger than the first type. It not only implies the first type and thus ensures the ability to distinguish from the global null hypothesis, but allows one to correctly characterize the dependency---that is, identify where and how the underlying distribution deviates from independence.
\begin{thm}[Local testing consistency]
\label{thm:local_consistency}
Suppose $X\not\independent Y$, and we observe i.i.d.\ pairs $(X_i,Y_i)$ from their joint distribution $F$. Let $A\in \A^{i,j}$ for some $i,j$ such that $F(A_{\cdot 0}),F(A_{\cdot 1}),F(A_{0\cdot}),F(A_{1\cdot})$ are all non-zero and $\theta(A) \neq 0$. Suppose either
\begin{itemize}
\item[(i)] $k_1$, $k_2$ and $M$ are fixed but large enough that $i\leq k_1$, $j\leq k_2$, and $i+j<M-1$,
\end{itemize}
or
\begin{itemize}
\item[(ii)] $k_1,k_2,M \rightarrow \infty$ such that they are $O(\log n)$.
\end{itemize}
Then we have as $n\rightarrow \infty$
\[
{\rm P}(p(A) < \alpha(A)\,|\,\bn_{k_1,0},\bn_{k_2,0}) \rightarrow 1 \quad \text{$F^{\infty}$-a.s.}
\]
and
\[
{\rm P}(p(A) < \alpha(A)) \rightarrow 1
\]
where $\alpha(A) = 1-(1-\alpha)^{1/(M+1)\cdot 1/T(r)\cdot 1/L(i,j)}$ is the window-specific adjusted level-$\alpha$ threshold under the three-stage \v{S}id\'{a}k correction as defined before. 
\end{thm}
\noindent Remark: 

In particular, the theorem guarantees the consistency of our recommended choice of $k_1$, $k_2$, and $M$ based on $n$.
\vspace{0.5em}

An immediate implication of the local, window-specific consistency is the global consistency that ensures the rejection of the joint null $X\independent Y$ under any alternative  $X\not \independent Y$.
\begin{thm}[Global testing consistency]
\label{thm:global_consistency}
Suppose $X\not\independent Y$, and either
\begin{itemize}
\item[(i)] $k_1$, $k_2$, and $M$ are fixed but large enough that there exists at least one $A\in\A^{i,j}$ with $i< k_1$, $j< k_2$, and $i+j<M-1$ such that $F(A_{\cdot 0}),F(A_{\cdot 1}),F(A_{0\cdot}),F(A_{1\cdot})$ are all non-zero and $\theta(A)\neq 0$;
\end{itemize}
or
\begin{itemize}
\item[(ii)] $k_1,k_2,M \rightarrow \infty$ as $n\rightarrow \infty$ such that they are $O(\log n)$.
\end{itemize}
Then we have as $n\rightarrow \infty$
\[
{\rm P}(p_{overall}<\alpha\,|\,\bn_{k_1,0},\bn_{0,k_2})\rightarrow 1 \quad \text{$F^{\infty}$-a.s.}
\]
and
\[
{\rm P}(p_{overall}<\alpha) \rightarrow 1.
\]
\end{thm}
\vspace{-1em}

\section{Numerical examples}
\label{sec:num_exam}
\subsection{Power study}
\label{sec:simulations}
Next we carry out a numerical study to evaluate the performance of FES and compare it both quantitatively and qualitatively to several state-of-the-art methods for testing variable dependency---namely, the distance correlation (dCor) \citep{szekely:2007,szekely:2009}, the maximal information coefficient (MIC) \citep{reshef:2011}, the mutual information statistic estimated based on $k$-nearest neighbors (MI-KNN) \citep{kraskov:2004,kinney:2014}, as well as three more classical test statistics for dependency---namely, Pearson's correlation, Hoeffding's $D$-statistic \citep{hoeffding:1948}, and a generalization of Fisher's exact test to $R\times C$ contingency tables based on the tail probabilities of multivariate hypergeometric distributions \citep{mehta:1986}. To make the classical Fisher exact test for $R\times C$ tables comparable to our FES method, we set $R=2^{k_1}$ and $C=2^{k_2}$. We first carry out a power study to evaluate the statistical performance of the different methods. 

In evaluating the statistical power, we consider a total of six signature dependency scenarios that are chosen to be representative of a wide variety of dependency structures. Five of the dependency settings---namely, linear, sine, circular, parabolic, and checkerboard---have been widely adopted in recent works for evaluating metrics of variable dependency \citep{reshef:2011,kinney:2014,filippi:2015}. We include one additional scenario which we believe is also very important in modern applications and especially ``big data'' settings, and that is when the dependency is local---involving only a small portion of the observations/probability mass.  \ref{fig:data_sample_scenarios} presents a realization of the six scenarios (at a small enough noise level that the patterns are clearly visible). The specific simulation settings are summarized in \ref{tab:simulation_setting}. We simulate from each of the six scenarios at 20 different noise levels ranging from low to high (1 to 20).
The sample size and the noise variance are chosen so that the power of the methods cover almost the whole range of $(0,1)$.

\begin{table}[h]
\begin{center}
\begin{tabular}{c|c|c}
Scenario & \# of data points & Simulation setting\\
\hline
&&\\
Linear & $300$ &  $X=U$ and $Y=X + 3\epsilon$ \\
&&\\
Sine & 300 & $X=U$ and $Y=\sin(5\pi X) + 4\epsilon$\\
&&\\
Circular & 300 & $X=\cos(\theta)+\epsilon$ and $Y=\sin(\theta)+\epsilon'$\\
&&\\
Parabolic & $300$ & $X=U$ and $Y=(X-0.5)^2 + 0.75\epsilon$\\
&&\\
Checkerboard & 500 & $X=W+\epsilon$ and $Y=\begin{cases}V_1 + \epsilon' & \text{if $W$ is odd}\\ V_2 + \epsilon' & \text{if $W$ is even} \end{cases}$  \\
&&\\
Local & 1000& $X=\epsilon$ and $Y=\begin{cases} X + 0.25\epsilon''  & \text{if $0\leq \epsilon,\epsilon'\leq 0.7$}\\ \epsilon' & \text{otherwise}\end{cases}$\\
&& \\
\hline
\end{tabular}
\end{center}
\caption{Six simulation scenarios. At each noise level $l=1,2,\ldots,20$, $\epsilon,\epsilon',\epsilon''\iid {\rm N}(0,(l/20)^2)$, and the following random variables are all independent: $U\sim {\rm Uniform}(0,1)$, $\theta\sim {\rm Unif}(-\pi,\pi)$, $W\sim$ Multi-Bern$(\{1,2,3\},(1/3,1/3,1/3))$, $V_1\sim$ Multi-Bern$(\{1,3,5\},(1/3,1/3,1/3))$, and $V_2\sim$ Bern$(\{2,4\},(1/2,1/2))$.}
\label{tab:simulation_setting}
\end{table}%

\begin{figure}[h]
\begin{center}
\includegraphics[width=30em]{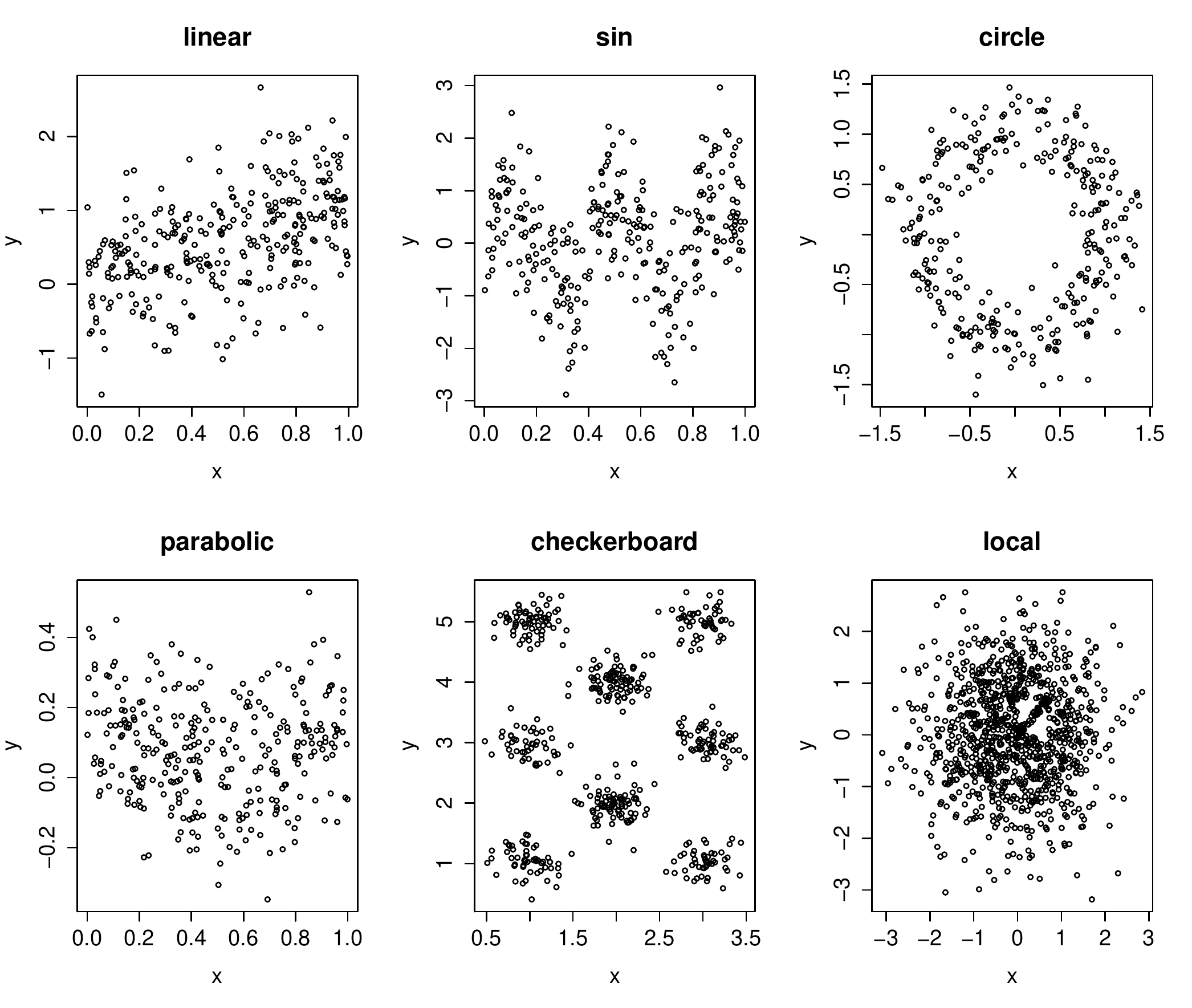}
\caption{Illustration of the six simulation scenarios. This gives an example simulation draw for the six scenarios at noise level $l=2$.}
\label{fig:data_sample_scenarios}
\end{center}
\vspace{-1em}
\end{figure}

We carry out 10,000 simulations under each setting and noise level, and estimate the power of six different methods---FES, dCor, MIC, Hoeffding's $D$ test, MI-KNN (with $k=10$), and Pearson's correlation ($R^2$). For FES, we adopt a screening rule that a window must contain at least $s=25$ observations and each row and column in the $2\times 2$ subtable must contain at least $s'=10$ data points. Accordingly, following our suggestion in Section~\ref{sec:fes}, we set $k_1=k_2=M+1=\lfloor \log_2(n/10)\rfloor$. 
The power as a function of the noise level for each of the methods is presented in \ref{fig:power_functions}.
\begin{figure}[ph!]
\begin{center}
\includegraphics[width=38em]{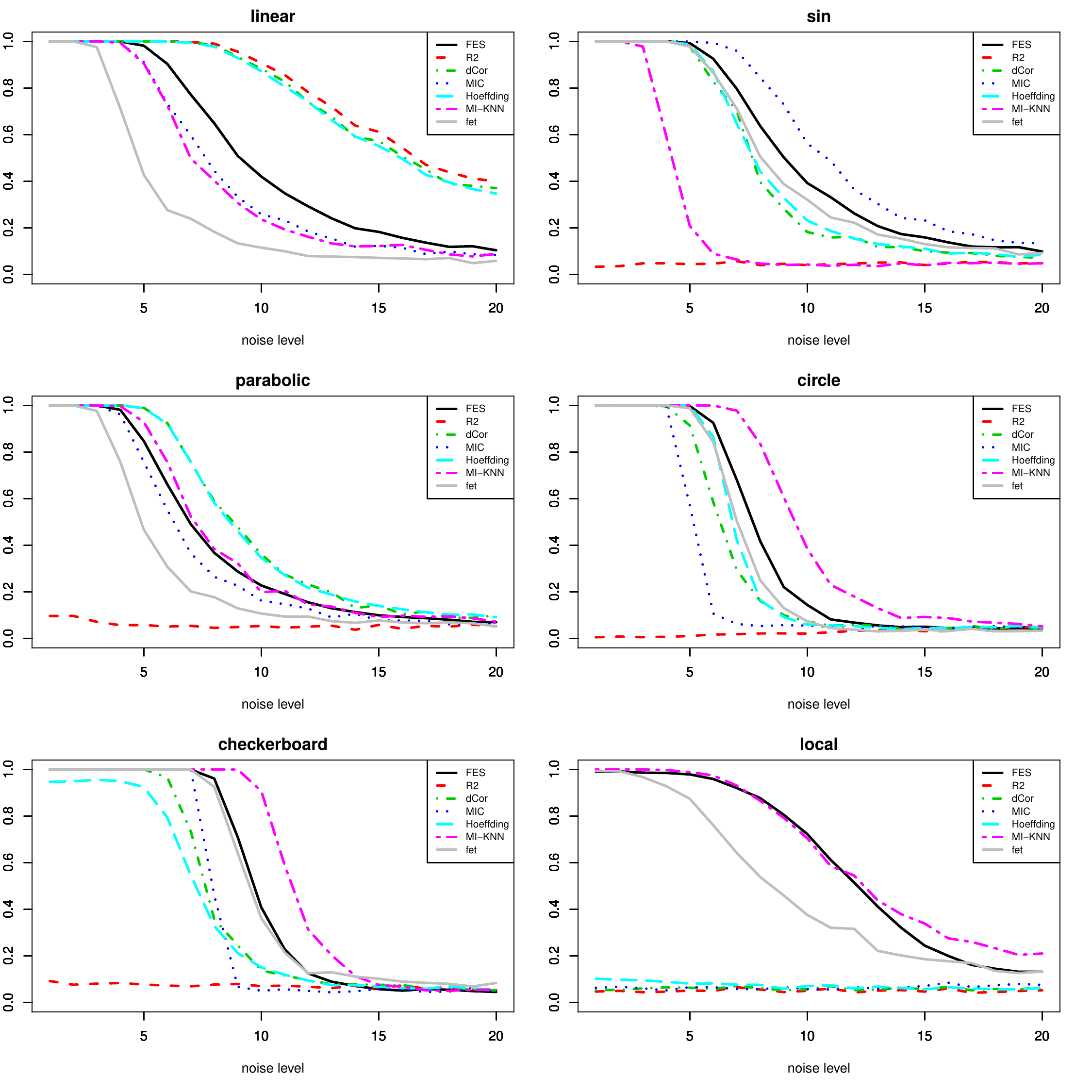}
\caption{Power under the six scenarios at 20 noise levels. Seven methods are compared: FES using three-stage exact \v{S}id\'{a}k's correction, Pearson's correlation ($R^2$), distance correlation (dCor), maximal information coefficient (MIC), Hoeffding's $D$ test, $k$-nearest neighbor based mutual information (MI-KNN with $k=10$), and a generalization of Fisher exact test to $R\times C$ tables (fet). The significance thresholds for all methods except FES and Hoeffding's $D$ are computed through simulation---standard Monte Carlo for fet and permutation for the rest.}
\label{fig:power_functions}
\end{center}
\vspace{-1em}
\end{figure}

In summary, no methods uniformly dominate all else under all simulation settings. FES behaves robustly across scenarios. Specifically, when pitched against each of the other methods, FES outperforms each competitor in at least as many scenarios as those in which it underperforms. This is consistent with earlier theoretical results that show that multi-scale scan tests using axis-aligned rectangles enjoy minimax optimality \citep{walther:2010}. Also, it is interesting to note that Pearson's correlation, dCor, MIC, and Hoeffding essentially lose all power in the ``local'' scenario when the dependency involves only a small subset of the observations. FES and MI-KNN, which do measure local features of the joint distribution, are the most powerful in such cases as expected. 

We also verify that the methods properly control the FWER, for otherwise the power comparison is not meaningful. To this end, we also carry out 10,000 simulations under a ``null'' scenario, under which $X$ and $Y$ are independent standard normal variables. (Note that FES is invariant to marginal transformations on the data, and hence its behavior under independence is not affected by the choice of the marginal distributions of $X$ and $Y$ at all.) We carry out this null simulation under 20 different sample sizes, $100, 200, \ldots, 2,000$, and for each simulation we applied the six methods compared previously. 

\ref{fig:null_power} presents the estimated power under the null (i.e., the FWER) as a function of the sample size for all methods. Except for FES and Hoeffding's $D$, all other methods are based on simulation---standard Monte Carlo for the classical Fisher's exact test on $R\times C$ tables and permutation for all other methods. Indeed we see that while the three-stage \v{S}id\'{a}k's correction controls the FWER at 5\%, FES with exact adjustment tends to be conservative like Fisher's exact test. The extent of conservativeness is not large after the mid-$p$ value correction, with the estimated FWER typically around 4\% across sample sizes. 
\begin{figure}[th]
\begin{center}
\includegraphics[clip=TRUE, trim=0 0 0 8mm, width=26em]{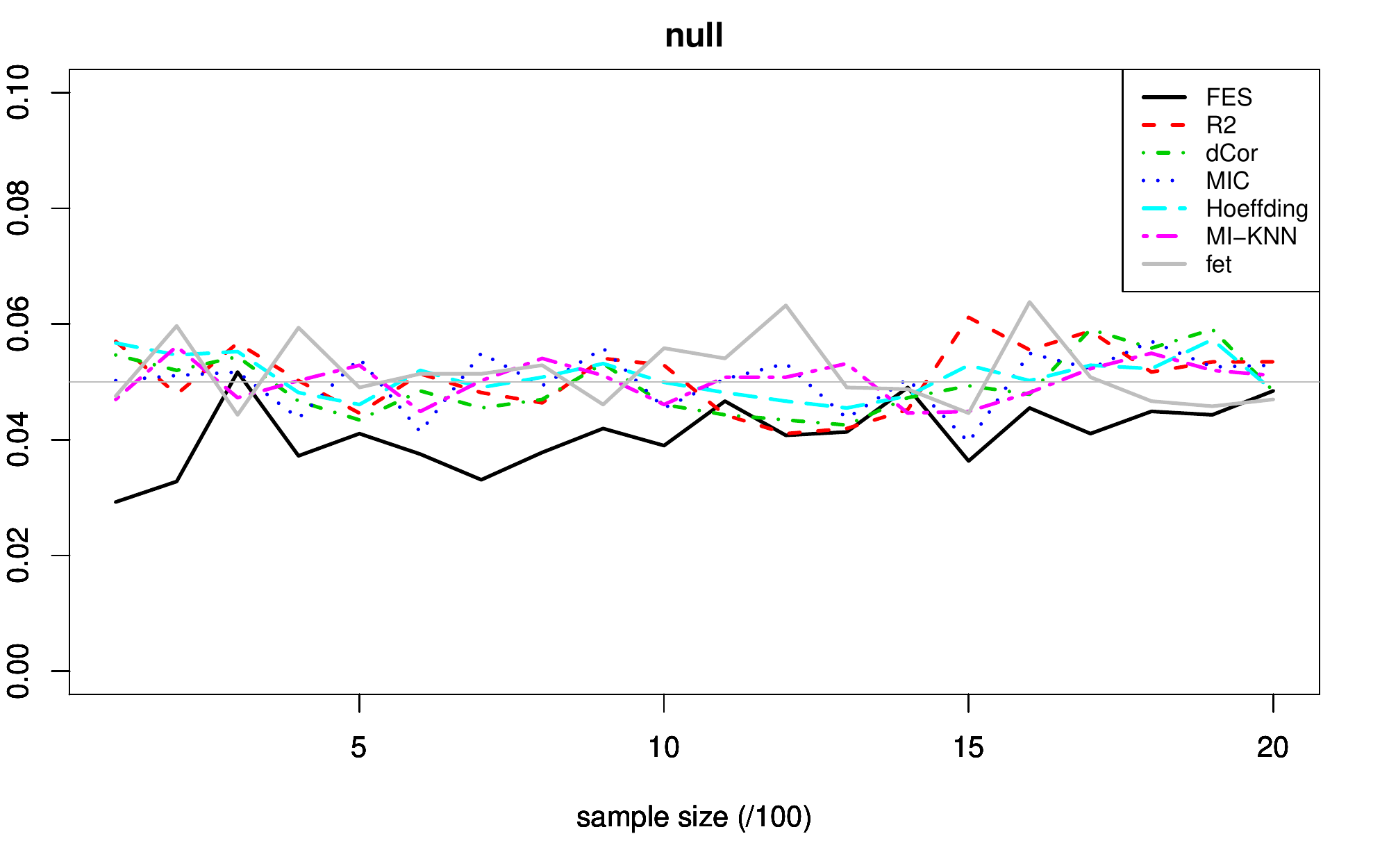}
\caption{Estimated FWER versus sample size for seven methods. The horizontal line marks 5\%---the level at which each method is aimed to control the FWER.}
\label{fig:null_power}
\end{center}
\vspace{-1em}
\end{figure}

In Supplementary Materials~S2, we present the results of two additional simulation studies. One is for comparing the seven methods under the six dependency scenarios at 20 different sample sizes, and the other investigates the sensitivity of FES to the choice of the resolution parameters $k_1$, $k_2$, and $M$.

\begin{figure}[p]
\begin{center}
\includegraphics[width=35em]{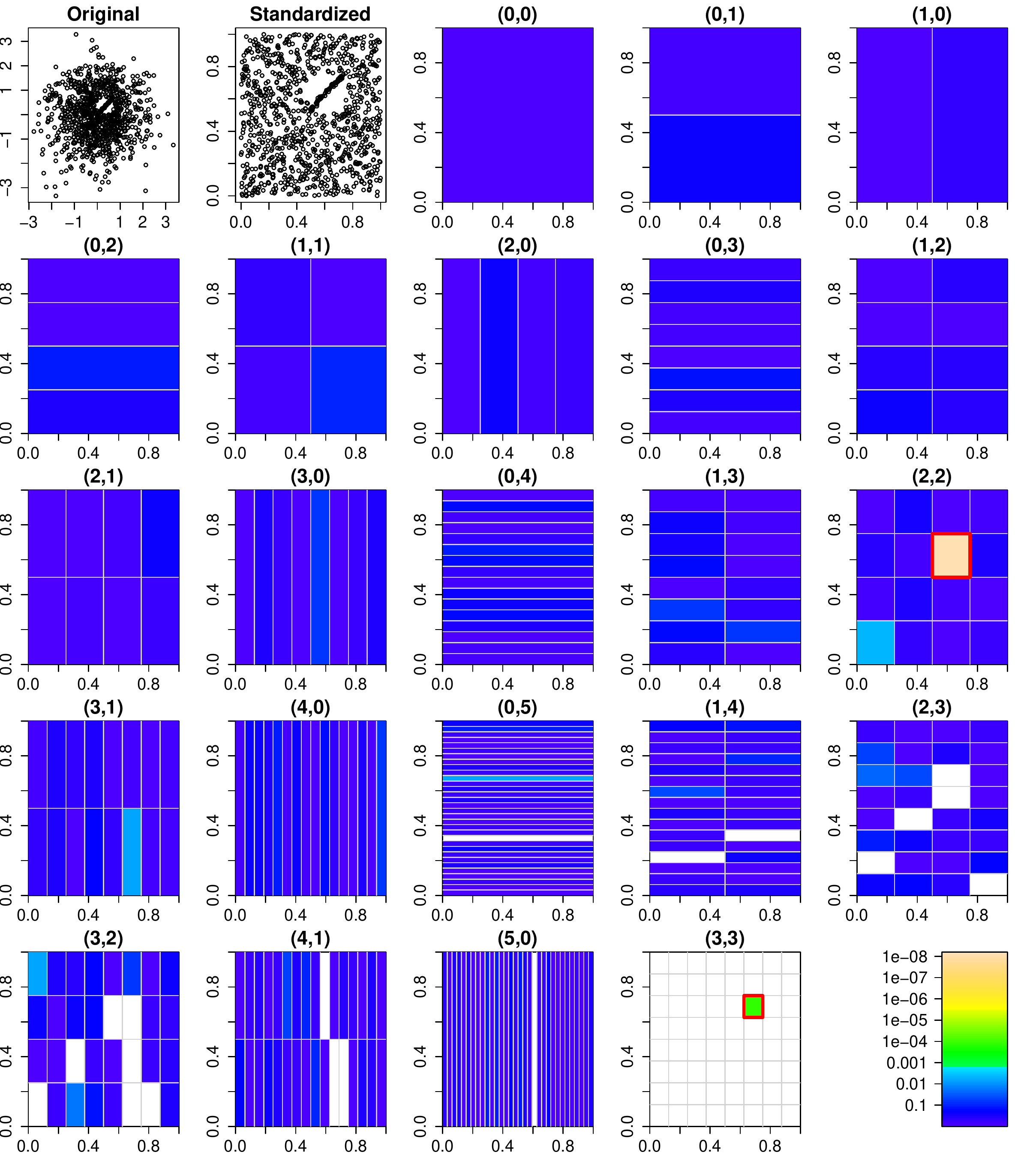}
\caption{Nominal $p$-values for the Fisher's exact test on each window scanned under FES on a sample drawn from the ``local'' scenario. The first two plots in the first row show the original data as well as the transformed data after empirical CDF transformation to the two margins. The last plot in the last row shows the color scale for the $p$-values. The other plots show the $p$-values for each stratum $\A^{i,j}$---with $(i,j)$ marked on top of each plot---that has at least one window passing the screening. White windows are those that have failed screening and so no $p$-values have been computed on them. The windows with red boundaries are deemed significant at the 5\% level---they have $p$-values smaller than the corrected 5\% threshold, i.e., $\alpha(i,j)$---under the three-stage \v{S}id\'{a}k's correction. In this scan, we have adopted $k_1=k_2=M+1=7$. The screening rule is that each column and row must have at least 10 observations with the whole window containing at least 25 data points.}
\label{fig:identifying_dependency}
\end{center}
\vspace{-1em}
\end{figure}

\vspace{-0.5em}

\subsection{Identifying local dependency}
Among the aforementioned methods, FES enjoys a unique ability to identify the structure of the underlying dependency, especially when it is local. As an illustration, we apply FES to a simulated sample under the ``local'' scenario. \ref{fig:identifying_dependency} presents the $p$-values from all scanning windows that have passed screening. (As we apply FES after transforming the marginals using the empirical CDFs, the windows are on the scale of the empirical quantiles. Thus we also present the standardized data in the figure.) Two windows, one in $\A^{2,2}$ and the other in $\A^{3,3}$ have $p$-values that are less than the respective \v{S}id\'{a}k adjusted critical threshold $\alpha(2,2)$ and $\alpha(3,3)$ and hence are identified as significant. These two windows indeed cover the actual portion of the sample space where the local dependency exists. They are marked using red boundaries.  A number of other windows, though having nominal Fisher's exact $p$-value less than 1\%, are not deemed significant after multiple testing adjustment. 
\vspace{-0.5em}

\subsection{Computational scalability}
The next comparison we make is in computational efficiency, and in particular the ability to handle large data sets. Huge sample sizes are commonplace in modern applications and are necessary for identifying weak or local dependencies. It is thus of interest to see how FES and the existing methods scale with the sample size in terms of computational demands. As such, we report the typical CPU time of a single run of each method {\em without permutation} as a function of sample size in \ref{fig:computing_time}. 

\begin{figure}[!ht]
\begin{center}
\includegraphics[width=38em]{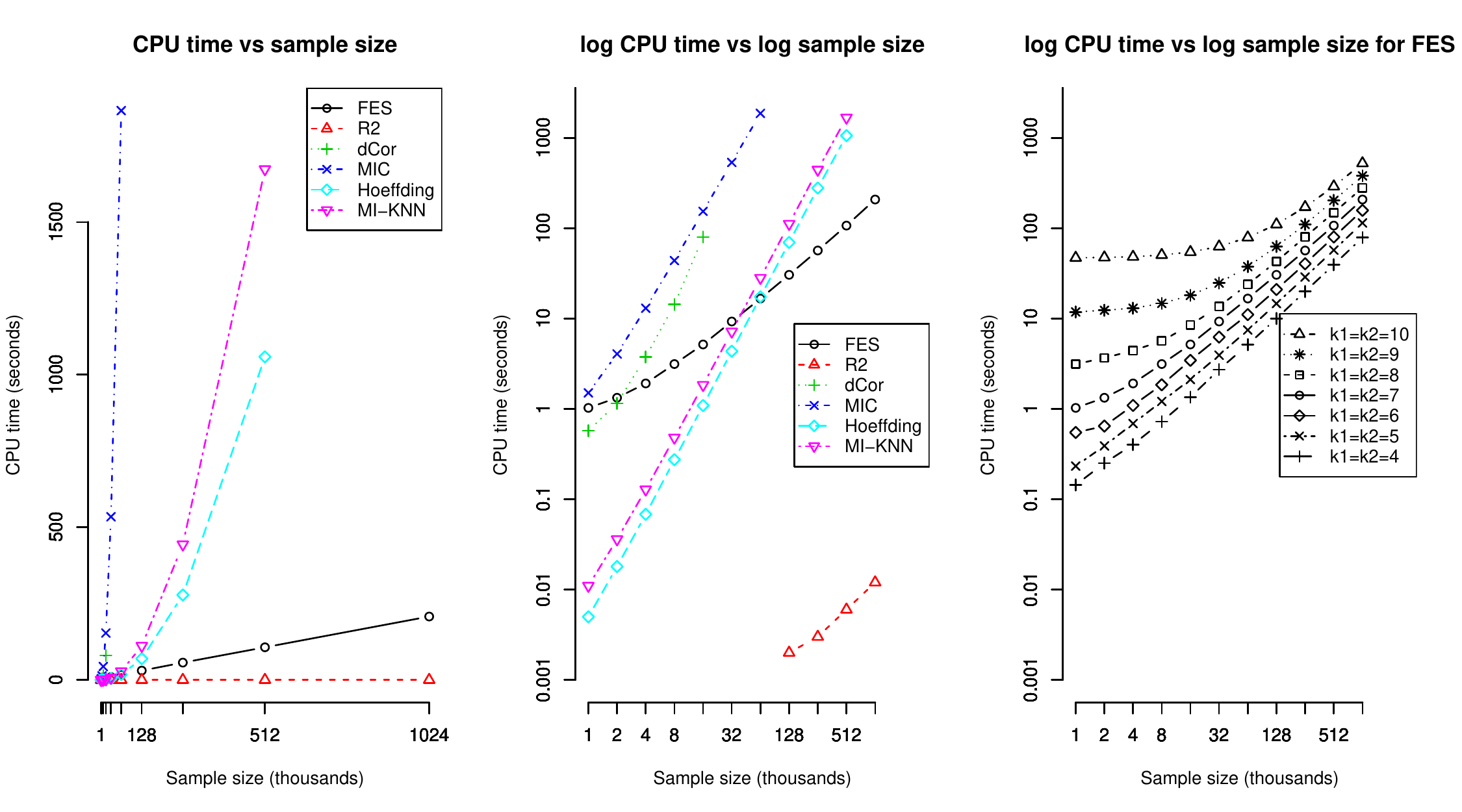}
\caption{Computing time (in seconds) of a single run without permutation vs sample size (in thousands). Left: Six methods in original scale, with $k_1=k_2=7$ for FES. Middle: Six methods in log-log scale, with $k_1=k_2=7$ for FES. Right: FES for three different choices of $k_1$ and $k_2$---from 4 to 10---in log-log scale. In the left and middle plots, FES is configured with $k_1=k_2=7$. In all runs, $M+1=\lfloor \log_2(n/16) \rfloor$ for FES. The sample sizes are $2^i\times 1000$ for $i=0,1,2,\ldots,10$. Each method is measured up to $i=10$, i.e., a sample size of $1.024$ million, or the maximum sample size for which computing is under 2,000 seconds, whichever is smaller, except for dCor, which is evaluated up to the sample size of 16,000 as the larger samples require more RAM than is available (32 Gbs) on our desktop.}
\label{fig:computing_time}
\end{center}
\vspace{-0.5em}
\end{figure}

In particular, 
FES and Pearson's correlation scale approximately linearly with the sample size, MIC, Hoeffding's $D$, and MI-KNN scale quadratically with sample size. The dCor also has quadratic complexity though it is hard to see from our figure. (Though the {\tt R} implementation of dCor, which we adopt here, has quadratic complexity in sample size, we note that there is a recent work \citep{huo:2016} proposes a new algorithm for computing dCor that has computational complexity $O(n\log n)$ in sample size.)

To see how the computation under FES scale with the marginal maximum resolution $k_1$ and $k_2$, we repeat the analysis for $k_1=k_2=4,5,\ldots,10$ respectively, and report the computing time in the right panel of \ref{fig:computing_time}. Different $(k_1,k_2)$ values affect the constant factor in the complexity but not the linear complexity itself, and the constant approximately grows by a factor of about $1.4$ for each simultaneous unit increment in $k_1$ and $k_2$. We note that while the actual scanning in FES for any fixed resolution specification is $O(n)$ in sample size, the optional preprocessing rank transform step recommended in FES has complexity $O(n\log n)$, although for the investigated sample sizes the computing time is dominated by the scanning in FES.

Again, we note that for FES, due to the exact inference recipe, a single run is sufficient, whereas for the methods compared here except Hoeffding's $D$, permutation is needed to properly control FWER. This makes FES even more attractive for data sets of massive sample sizes. 

While the comparison focuses on CPU time, FES uses constant memory for each combination of $k_1$, $k_2$, and $M$, regardless of the sample size (aside from the memory that is required for storing the data), and so RAM is not a concern in applying FES to massive data sets. The above simulations required less than 500 Mbs of RAM for FES.

\section{Application to the American Gut microbiome data}
\label{sec:data_exam}

The human microbiome is the community of numerous microbes that inhabit the human body. Understanding the microbiome can provide insights into various aspects of human health. Microbiome data is often presented in the form of OTU (Operational Taxonomic Unit, which could be viewed as pragmatic proxies for "species") tables, which consist of counts of various OTUs in a number of microbiome samples. A common task in analyzing microbiome data is evaluating pairwise dependency in OTU relative abundance \citep{Mandal:2015aa, reshef:2011}. (The relative abundance of an OTU in each sample is the proportion of counts among all counts for that sample.)

We apply FES to detect statistically significant dependency in relative abundance among OTU pairs in a data set from the American Gut Project \citep{mcdonald:2015}. The project collects fecal, oral, skin, and other body site microbiome samples from a large number of participants. 
The OTU table being analyzed comes from the July 29, 2016 version of the fecal data which contains the counts of $27774$ OTUs.
The data are freely available to the public. Although the total number of OTUs in a typical sample is huge, the OTU table is very sparse---with most OTUs having essentially no counts from all but a very small number of samples. In this illustration, we analyze the $100$ OTUs with the largest overall counts across samples, use the samples with at most $15$ zero counts in the top $100$ OTUs ($n=514$). These top $100$ OTUs contain about $2/3$ of the total counts in the OTU table. 

Instead of simply presenting a list of most significantly dependent OTU pairs, we investigate how FES behaves in relation to commonly adopted metrics for measuring dependency. Our motivation is simple---if popular metrics such as MIC and dCor give a numeric score that quantifies the extent of dependency without directly providing an evaluation on the statistical significance (not without resampling), then it would be interest to see whether the $p$-values produced from FES give roughly consistent ranking of the OTU pairwise dependency with respect to the dependency metric. (The consistency cannot be perfect as our power study shows.) If this is the case, then one can in fact use FES in combination with the corresponding metric---with the latter giving an overall summary of the extent of dependency and FES providing a quick, resampling-free evaluation on the statistical significance.

To this end, we consider four popular dependency metrics---MIC, dCor, KNN-MI, $R^2$. We rank all OTU pairs in terms of the FES $p$-value (with three-stage \v{S}id\'{a}k's correction),
and plot that ranking versus each of the four metrics (\ref{fig:compare_plot_fes_only}). A strict monotone decreasing pattern will correspond to a perfect consistency between FES and the metric. 

\begin{figure}[!t]
\begin{center}
\includegraphics[width=40em]{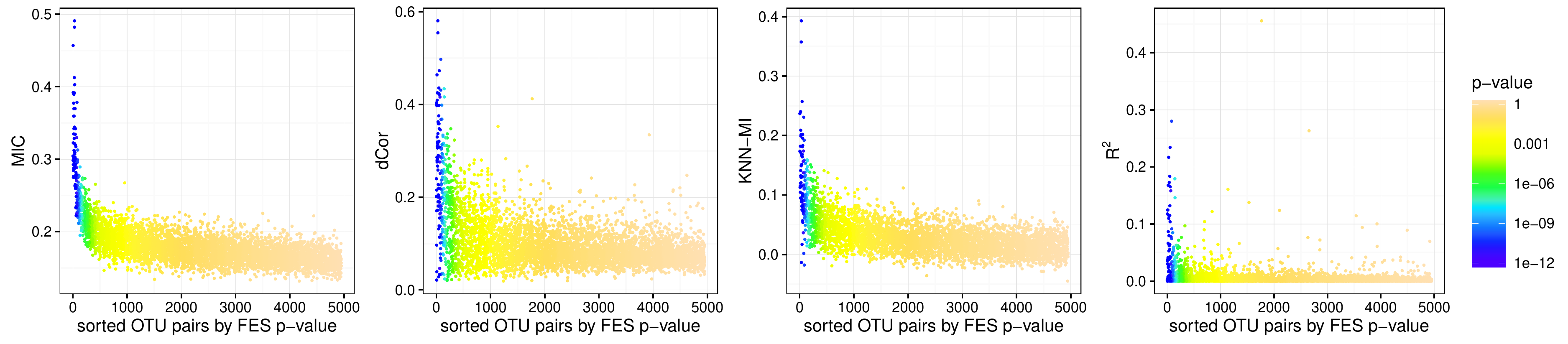}
\caption{FES $p$-value ranking versus four dependency metrics. In each subplot, OTU pairs are sorted and colored by the FES p-value. The $x$-axis shows the ranking of each OTU pair, the $y$-axis shows one of the four dependency measures (MIC, dCor, KNN-MI, $R^2$). The dashed vertical lines mark the nominal $5\%$ $p$-value cutoffs.}
\label{fig:compare_plot_fes_only}
\end{center}
\end{figure}

\begin{figure}[!h]
\begin{center}
\includegraphics[width=36em]{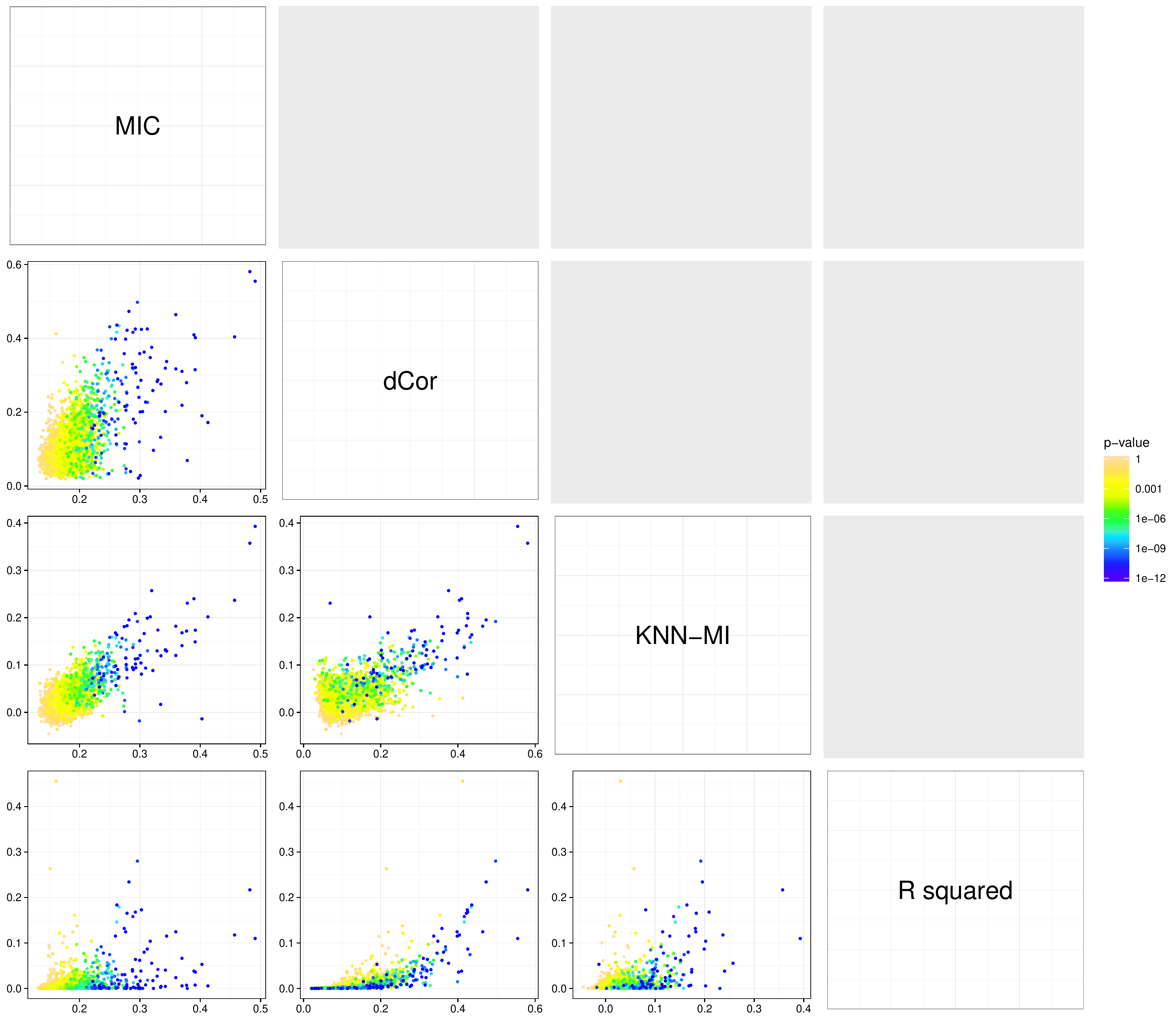}
\caption{Pairwise comparison for four dependency metrics with color coded by FES p-value.}
\label{fig:matrix_plot_lower}
\end{center}
\vspace{-1em}
\end{figure}

\ref{fig:compare_plot_fes_only} shows that the statistical significance ranking rendered by FES is most consistent with the two information theoretic metrics MIC and KNN-MI. Also, note that most pairs that have very small FES $p$-values have Pearson's correlation closed to zero, indicating that dependency among OTU pairs is generally non-linear. 

To investigate how the four dependency metrics are consistent among each other, we plot every metric against every other with the color again determined by the FES $p$-value (\ref{fig:matrix_plot_lower}). The two information theoretic metrics, MIC and KNN-MI, show the most consistent pattern with each other. On the other hand, dCor appears to show stronger consistency with KNN-MI than with MIC. The consistency between Pearson's correlation with each of the other metrics is weak, suggesting that each metric is capable to characterizing non-linear dependency in their own ways.

\begin{figure}[!th]
\begin{center}
\includegraphics[width=35em]{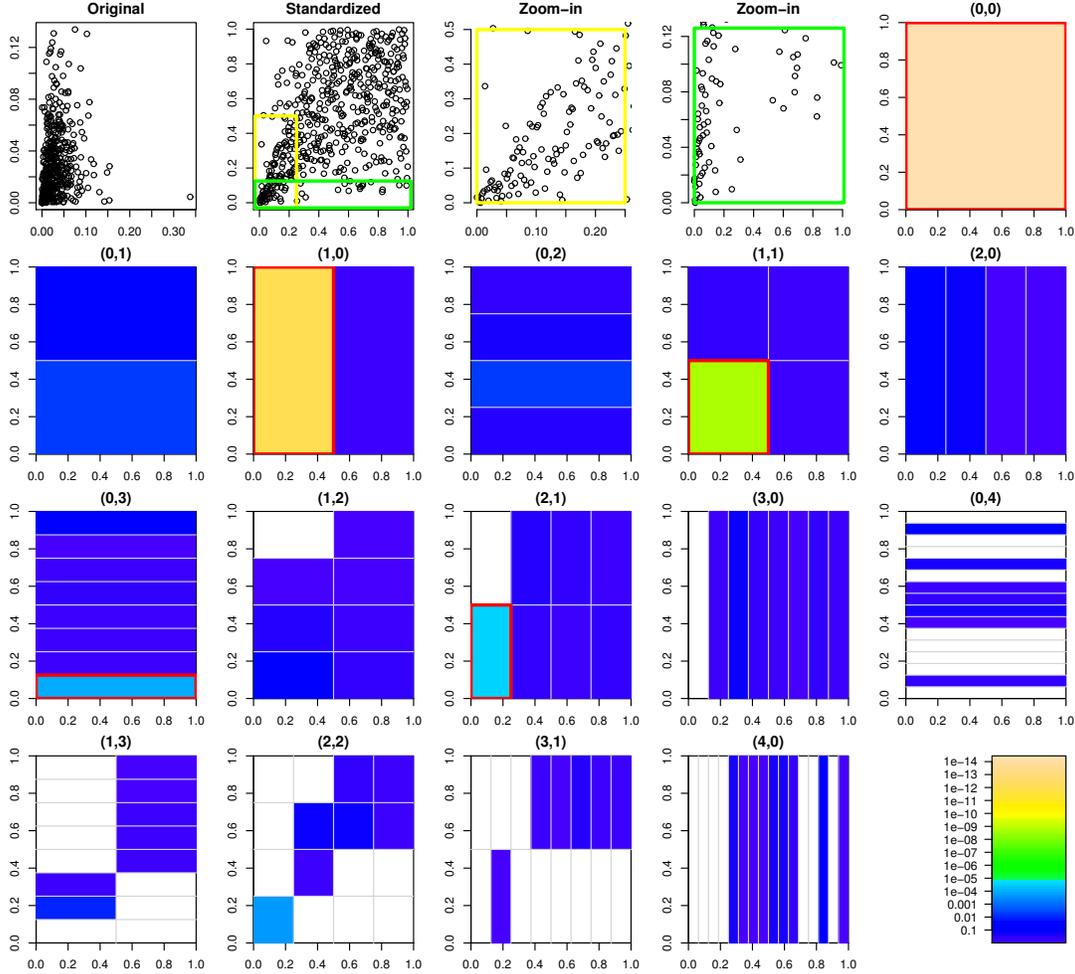}
\caption{Nominal $p$-values for the Fisher's exact test on each window scanned under FES on a certain pair of OTU samples. The first two plots in the first row show the original data as well as the transformed data after empirical CDF transformation to the two margins. The third and the fourth plot in the first row show the zoom-in scatter plot of the transformed samples corresponding to two windows in resolution level 3 that are deemed significant at the $5\%$ level by FES after the three-stage \v{S}id\'{a}k's correction. The other plots show the $p$-values for each stratum with the red rectangles indicating the significant windows. The overall p-value under the three-stage \v{S}id\'{a}k's correction is  $1.465\times 10^{-14}$. }
\label{fig:local}
\end{center}
\end{figure}

There are many indications of local dependency---which we defined empirically as the case when particularly strong evidence of dependency is localized in a subset of the observations---among the OTUs. 
We illustrate this in an example of OTU pairs (OTU 4481131 and OTU 4478125) in which there is strong evidence of dependency among individuals with low levels of abundance in both of these two bacteria, while among those individuals with high abundance in one or both OTUs, such dependency disappears. \ref{fig:local} shows the FES scanning plots for this pair of OTUs, and we note those highly significant windows in the lower left corner of the sample space. One possible explanation for the local dependency is that these two OTUs are functionally highly similar, and therefore responds similarly to the environmental/dietary conditions in which each can grow. As such, individuals with very low levels of one OTU tend to have low levels of the other as well. On the other hand, when the environment or dietary conditions are conducive to growth for these two species, they function largely independently and therefore do not display dependency in cases where one or both are highly abundant. This conjecture is supported through verifying their evolutionary relationship on the phylogenetic tree. It turns out that they have the same taxonomic ranks, and belong to the same subspecies \textit{prausnitzii} under the genus of {\it Faecalibacterium}.

\section{Concluding remarks}
\label{sec:discussion}

We have introduced Fisher exact scanning as a method for testing and identifying dependency conditional on the marginals of the observations in a fashion that generalizes Fisher's exact test on $2\times 2$ tables. We have showed that FES performs well in a variety of non-linear dependency settings, and is particularly powerful for identifying local dependencies. Moreover, its simple statistical properties allow exact inference without resorting to resampling. This, along with its linear computational scalability, makes it a desirable method for handling data sets with lots of observations. Moreover, not only does it allow the test of the null hypothesis of independence, but the identification of the nature of the dependency as well. 

We have mainly concerned ourselves in this work regarding evaluating statistical significance in testing dependence. In practice, one is often also interested in evaluating the scientific significance though measuring the extent of dependency. To this end, one can also report the empirical odds-ratios or a confidence interval for the OR on the windows identified as significant under FES. Alternatively, one can apply FES in conjunction with popular metrics of dependency as we investigated in \ref{sec:data_exam}. In this regard, our numerical results suggest that FES provides decent ranking for information theoretic metrics such as MIC and KNN-MI. Thus one could use FES to identify {\em statistically} significant variable pairs and use MIC or KNN-MI as a numeric score for measuring the dependency relationship.  This avoids the computational burden in resampling-based significance evaluation for big data.

Finally, we believe that the factorization of the MHG likelihood on $R\times C$ contingency tables into a product of HG likelihoods has further applications beyond what is exploited here, and is worth further investigation on its own right. 

\section*{Software}
For MIC, we use the {\tt MINE} application downloaded from {\tt http://www.exploredata.net/}. For KNN-MI, we use the {\tt mutinfo} function in the {\tt R} package {\tt FNN}. For dCor, we use the {\tt dcor} function in the {\tt R} package {\tt energy}. For Hoeffding's $D$, we use the {\tt hoeffd} function in the {\tt R} package {\tt Hmisc}. For the classical Fisher exact test for $R\times C$ tables, we use {\tt fisher.test} in {\tt R}. Our {\tt R} package {\tt FES} is freely available at {\tt https://github.com/MaStatLab/FES}.

\section*{Acknowledgment}
This research is partly supported by NSF grant DMS-1612889 and a Google Faculty Research Award. 

\bibliography{FES}

\begin{thebibliography}{}

\bibitem[\protect\citeauthoryear{Agresti}{Agresti}{2013}]{agresti2013categorical}
Agresti, A. (2013).
\newblock {\em Categorical Data Analysis}.
\newblock Wiley Series in Probability and Statistics. Wiley.

\bibitem[\protect\citeauthoryear{Choi, Blume, and Dupont}{Choi
  et~al.}{2015}]{choi:2015}
Choi, L., J.~D. Blume, and W.~D. Dupont (2015).
\newblock Elucidating the foundations of statistical inference with $2\times 2$
  tables.
\newblock {\em PLOS ONE\/}~{\em 10\/}(4).

\bibitem[\protect\citeauthoryear{{Filippi} and {Holmes}}{{Filippi} and
  {Holmes}}{2015}]{filippi:2015}
{Filippi}, S. and C.~{Holmes} (2015).
\newblock {A Bayesian nonparametric approach to testing for dependence between
  random variables}.
\newblock {\em ArXiv e-prints\/}.

\bibitem[\protect\citeauthoryear{Fisher}{Fisher}{1954}]{fisher:1954}
Fisher, R.~A. (1954).
\newblock {\em {Statistical methods for research workers}}.
\newblock Edinburgh: Oliver and Boyd.

\bibitem[\protect\citeauthoryear{Gretton, Fukumizu, Teo, Song, Sch\"{o}lkopf,
  and Smola}{Gretton et~al.}{2008}]{gretton:2007}
Gretton, A., K.~Fukumizu, C.~H. Teo, L.~Song, P.~B. Sch\"{o}lkopf, and A.~J.
  Smola (2008).
\newblock A kernel statistical test of independence.
\newblock In J.~C. Platt, D.~Koller, Y.~Singer, and S.~T. Roweis (Eds.), {\em
  Advances in Neural Information Processing Systems 20}, pp.\  585--592. Curran
  Associates, Inc.

\bibitem[\protect\citeauthoryear{Heller, Heller, and Gorfine}{Heller
  et~al.}{2013}]{heller:2012}
Heller, R., Y.~Heller, and M.~Gorfine (2013).
\newblock A consistent multivariate test of association based on ranks of
  distances.
\newblock {\em Biometrika\/}~{\em 100\/}(2), 503.

\bibitem[\protect\citeauthoryear{Heller, Heller, Kaufman, Brill, and
  Gorfine}{Heller et~al.}{2016}]{heller:2016}
Heller, R., Y.~Heller, S.~Kaufman, B.~Brill, and M.~Gorfine (2016, January).
\newblock Consistent distribution-free k-sample and independence tests for
  univariate random variables.
\newblock {\em J. Mach. Learn. Res.\/}~{\em 17\/}(1), 978--1031.

\bibitem[\protect\citeauthoryear{Hoeffding}{Hoeffding}{1948}]{hoeffding:1948}
Hoeffding, W. (1948).
\newblock A non-parametric test of independence.
\newblock {\em Ann. Math. Statist.\/}~{\em 19\/}(4), 546--557.

\bibitem[\protect\citeauthoryear{Huo and Sz\'ekely}{Huo and
  Sz\'ekely}{2016}]{huo:2016}
Huo, X. and G.~J. Sz\'ekely (2016).
\newblock Fast computing for distance covariance.
\newblock {\em Technometrics\/}~{\em 58\/}(4), 435--447.

\bibitem[\protect\citeauthoryear{Kinney and Atwal}{Kinney and
  Atwal}{2014}]{kinney:2014}
Kinney, J.~B. and G.~S. Atwal (2014).
\newblock Equitability, mutual information, and the maximal information
  coefficient.
\newblock {\em PNAS\/}~{\em 111\/}(9), 3354--3359.

\bibitem[\protect\citeauthoryear{Kolaczyk and Nowak}{Kolaczyk and
  Nowak}{2004}]{kolaczyk:2004}
Kolaczyk, E.~D. and R.~D. Nowak (2004).
\newblock Multiscale likelihood analysis and complexity penalized estimation.
\newblock {\em The Annals of Statistics\/}~{\em 32\/}(2), 500--527.

\bibitem[\protect\citeauthoryear{Kou and Ying}{Kou and Ying}{1996}]{kou:1996}
Kou, S.~G. and Z.~Ying (1996).
\newblock Asymptotics for a 2 × 2 table with fixed margins.
\newblock {\em Statistica Sinica\/}~{\em 6\/}(4), 809--829.

\bibitem[\protect\citeauthoryear{Kraskov, St\"ogbauer, and Grassberger}{Kraskov
  et~al.}{2004}]{kraskov:2004}
Kraskov, A., H.~St\"ogbauer, and P.~Grassberger (2004).
\newblock Estimating mutual information.
\newblock {\em Phys. Rev. E\/}~{\em 69}, 066138.

\bibitem[\protect\citeauthoryear{Lancaster}{Lancaster}{1949}]{lancaster:1949}
Lancaster, H.~O. (1949).
\newblock The derivation and partition of $\chi^2$ in certain discrete
  distributions.
\newblock {\em Biometrika\/}~{\em 36\/}(1/2), 117--129.

\bibitem[\protect\citeauthoryear{Little}{Little}{1989}]{little:1989}
Little, R. J.~A. (1989).
\newblock {Testing the equality of two independent binomial proportions}.
\newblock {\em Am Statistician\/}~{\em 43}, 283--288.

\bibitem[\protect\citeauthoryear{Ma}{Ma}{2016}]{ma:2016}
Ma, L. (2016).
\newblock Adaptive shrinkage in {P}\'olya tree type models.
\newblock {\em Bayesian Analysis\/}~{\em (In press.)}.

\bibitem[\protect\citeauthoryear{Mandal, Van~Treuren, White, Eggesb{\o},
  Knight, and Peddada}{Mandal et~al.}{2015}]{Mandal:2015aa}
Mandal, S., W.~Van~Treuren, R.~A. White, M.~Eggesb{\o}, R.~Knight, and S.~D.
  Peddada (2015).
\newblock Analysis of composition of microbiomes: a novel method for studying
  microbial composition.
\newblock {\em Microbial Ecology in Health and Disease\/}~{\em 26},
  10.3402/mehd.v26.27663.

\bibitem[\protect\citeauthoryear{McDonald, Hornig, Lozupone, Debelius, Gilbert,
  and Knight}{McDonald et~al.}{2015}]{mcdonald:2015}
McDonald, D., M.~Hornig, C.~Lozupone, J.~Debelius, J.~Gilbert, and R.~Knight
  (2015).
\newblock Towards large-cohort comparative studies to define the factors
  influencing the gut microbial community structure of asd patients.
\newblock {\em Microbial Ecology in Health and Disease\/}~{\em 26\/}(0).

\bibitem[\protect\citeauthoryear{Mehta and Patel}{Mehta and
  Patel}{1986}]{mehta:1986}
Mehta, C.~R. and N.~R. Patel (1986, June).
\newblock Algorithm 643: Fexact: A {F}ortran subroutine for {F}isher's exact
  test on unordered $r\times c$ contingency tables.
\newblock {\em ACM Trans. Math. Softw.\/}~{\em 12\/}(2), 154--161.

\bibitem[\protect\citeauthoryear{Reshef, Reshef, Finucane, Grossman, McVean,
  Turnbaugh, Lander, Mitzenmacher, and Sabeti}{Reshef
  et~al.}{2011}]{reshef:2011}
Reshef, D.~N., Y.~A. Reshef, H.~K. Finucane, S.~R. Grossman, G.~McVean, P.~J.
  Turnbaugh, E.~S. Lander, M.~Mitzenmacher, and P.~C. Sabeti (2011).
\newblock Detecting novel associations in large data sets.
\newblock {\em Science\/}~{\em 334\/}(6062), 1518--1524.

\bibitem[\protect\citeauthoryear{Resnick}{Resnick}{1999}]{resnick:1999}
Resnick, S.~I. (1999).
\newblock {\em A probability path\/} (2nd ed.).
\newblock Birkhauser.

\bibitem[\protect\citeauthoryear{Rufibach and Walther}{Rufibach and
  Walther}{2010}]{rufibach:2010}
Rufibach, K. and G.~Walther (2010).
\newblock The block criterion for multiscale inference about a density, with
  applications to other multiscale problems.
\newblock {\em Journal of Computational and Graphical Statistics\/}~{\em
  19\/}(1), 175--190.

\bibitem[\protect\citeauthoryear{Sz\'ekely and Rizzo}{Sz\'ekely and
  Rizzo}{2009}]{szekely:2009}
Sz\'ekely, G.~J. and M.~L. Rizzo (2009).
\newblock Brownian distance covariance.
\newblock {\em Ann. Appl. Stat.\/}~{\em 3\/}(4), 1236--1265.

\bibitem[\protect\citeauthoryear{Sz\'ekely, Rizzo, and Bakirov}{Sz\'ekely
  et~al.}{2007}]{szekely:2007}
Sz\'ekely, G.~J., M.~L. Rizzo, and N.~K. Bakirov (2007).
\newblock Measuring and testing dependence by correlation of distances.
\newblock {\em Ann. Statist.\/}~{\em 35\/}(6), 2769--2794.

\bibitem[\protect\citeauthoryear{Walther}{Walther}{2010}]{walther:2010}
Walther, G. (2010).
\newblock Optimal and fast detection of spatial clusters with scan statistics.
\newblock {\em The Annals of Statistics\/}~{\em 38\/}(2), 1010--1033.

\bibitem[\protect\citeauthoryear{{Wang}, {Jiang}, and {Liu}}{{Wang}
  et~al.}{2016}]{wang:2016}
{Wang}, X., B.~{Jiang}, and J.~S. {Liu} (2016, April).
\newblock {Generalized R-squared for Detecting Non-independence}.
\newblock {\em ArXiv e-prints\/}.

\bibitem[\protect\citeauthoryear{Zhang}{Zhang}{2017}]{zhang:2017}
Zhang, K. (2017).
\newblock {BET on independence}.

\end{thebibliography}

\newpage

\beginsupplement
\setcounter{page}{1}

\section*{Supplementary Materials}
\subsection*{S1.~Proofs}
\begin{proof}[Proof of Lemma~\ref{lem:equiv_k_indep}]
The necessity of the condition when $X\independentk Y$ follows immediately from Definition~\ref{defn:k1k2_indep}. The sufficiency follows because for any $l_1,l_2$,
\begin{align*}
F(I^{k_1}_{l_1}\times I^{k_2}_{l_2}) &= F(I^{k_1}_{l_1}\times I^{k_2}_{l_2})\sum_{l_1'}\sum_{l_2'} F(I^{k_1}_{l_1'}\times I^{k_2}_{l_2'})\\
&=\sum_{l_1'} F(I^{k_1}_{l_1'}\times I^{k_2}_{l_2}) \cdot \sum_{l_2'}  F(I^{k_1}_{l_1}\times I^{k_2}_{l_2'})=F_{X}(I^{k_1}_{l_1})F_{Y}(I^{k_2}_{l_2}).
\end{align*}
This completes the proof.
\end{proof}

\begin{proof}[Proof of Lemma~\ref{lem:k_independence}]
Without loss of generality, we just need to show that $(k_1,k_2)$-independence implies $(k_1-1,k_2)$-independence for $k_1\geq 1$. This follows immediately from the definition of $(k_1,k_2)$-independence and the fact that any $I^{k_1-1}_{l_1} = I^{k_1}_{2l_1-1}\cup I^{k_1}_{2l_1}$, and so
\begin{align*}
F(I^{k_1-1}_{l_1}\times I^{k_2}_{l_2}) &= F(I^{k_1}_{2l_1-1}\times I^{k_2}_{l_2}) + F(I^{k_1}_{2l_1}\times I^{k_2}_{l_2})= F_{X}(I^{k_1}_{2l_1-1}) F_{Y}(I^{k_2}_{l_2}) + F_{X}(I^{k_1}_{2l_1}) F_{Y}(I^{k_2}_{l_2})\\
&= F_{X}(I^{k_1-1}_{l_1}) F_{Y}(I^{k_2}_{l_2}).
\end{align*}
This completes the proof.
\end{proof}

\begin{proof}[Proof of Theorem~\ref{thm:independent_k}]
First, the fact that independence implies $(k_1,k_2)$-independence for all $k_1$ and $k_2$ follows immediately from the definition of the latter. To see the reverse, let $\mathcal{I}=\bigcup_{k=0}^{\infty}\mathcal{I}^k.$ For $A\in\mathcal{B}([0,1])$, let $[X\in A]=\{\omega: X(\omega)\in A\}$. We have $\sigma(\mathcal{I})=\mathcal{B}([0,1])$, where $\mathcal{B}([0,1])$ denotes the Borel $\sigma$-algebra on $[0,1]$. Let $\sigma(X)=\{[X\in A], A\in\mathcal{B}([0,1]) \}$, $\mathcal{C}_X=\{[X\in B], B\in\mathcal{I}\}$. We claim that $\sigma(\mathcal{C}_X)=\sigma(X)$. To see this, note that
$$
\sigma(\mathcal{C}_X)=\sigma(X^{-1}(B), B\in\mathcal{I})
=\sigma(X^{-1}(\mathcal{I}))
=X^{-1}(\sigma(\mathcal{I}))
=\sigma(X).
$$
Similarly, we could define $\sigma(Y), \mathcal{C}_Y$ and have $\sigma(\mathcal{C}_Y)=\sigma(Y)$. Following the definition of nested dyadic partition, we know that $\mathcal{C}_X, \mathcal{C}_Y$ are $\pi$-systems (with the empty set included).  Since $\mathcal{C}_X, \mathcal{C}_Y$ are independent classes, based on the Basic Criterion in \citep{resnick:1999} (page 92) we have that $\sigma(\mathcal{C}_X)$ and $\sigma(\mathcal{C}_Y)$ are independent $\sigma$-fields. Therefore, $\sigma(X)$ and $\sigma(Y)$ are independent, and thus $X$ and $Y$ are independent.
\end{proof}

\begin{proof}[Proof of Theorem~\ref{thm:independent_k_or}]
First, suppose $X\independentk Y$. 
For any $A\in \A^{k_1-1,k_2-1}$, $A_{00},A_{01},A_{10},A_{11}\in \A^{k_1,k_2}$ and so $\theta(A)=0$ by the definition of $(k_1,k_2)$-independence. Lemma~\ref{lem:k_independence} implies that $X$ and $Y$ are $(k_1',k_2')$-independent for all $0\leq k_1'\leq k_1-1$ and $0 \leq k_2' \leq k_2-1$, and so by the above reasoning for all such $(k_1',k_2')$, $\theta(A)=0$ for $A\in \A^{k_1'-1,k_2'-1}$. This proves that 
\[
X\independentk Y \quad \Rightarrow \quad \theta(A)=0 \quad \text{for all $A\in \A^{(k_1-1,k_2-1)}$}.
\]

To see the reverse, we first state and prove two propositions.

\begin{prop}
For a $2\times 4$ contingency table with cell probabilities $\pi_{11}=A, \pi_{12}=B,\cdots, \pi_{24}=H$ (see \ref{app_tab_1}), the following two sets of conditions are equivalent:
\begin{equation*}
\begin{aligned}
\mathcal{P}_1 &= \big\{AF=BE, CH=DG, BG=CF\big\} \\
\mathcal{P}_2 &= \big\{AF=BE, CH=DG, (A+B)(G+H)=(E+F)(C+D)\big\}
\end{aligned}
\end{equation*}
\label{prop:2by4table}
\end{prop}

\begin{table}[h]
\begin{center}
  \begin{tabular}{ | c  | c | c | c |}
    \hline
    $A$ & $B$  & $C$ & $D$ \\  \hline
    $E$ & $F$  & $G$ & $H$\\  \hline
  \end{tabular}
\end{center}
\caption{$2\times 4$ contingency table.}
\label{app_tab_1}
\end{table}

To see that $\mathcal{P}_1 \Rightarrow \mathcal{P}_2$, we only need to show that under $\mathcal{P}_1$, $AG+AH+BG+BH=CE+CF+DE+DF.$ 
Since $\frac{AF}{BE}=\frac{BG}{CF}=1$, we have $\frac{AF}{BE}\times\frac{BG}{CF}=\frac{AG}{CE}=1.$ Similarly, $\frac{BH}{DF}=\frac{AH}{DE}=1.$ Therefore, $AG+AH+BG+BH=CE+CF+DE+DF.$ To see that $\mathcal{P}_2 \Rightarrow \mathcal{P}_1$, We only need to show that $\mathcal{P}_2$ implies $BG=CF$. Here we use proof by contradiction. If $BG>CF$, since $\frac{AF}{BE}\times\frac{BG}{CF}=\frac{AG}{CE}$ and $\frac{AF}{BE}=1$, we have $AG>CE$. Since $\frac{BG}{CF}\times\frac{CH}{DG}=\frac{BH}{DF}$ and $\frac{CH}{DG}=1$, we have $BH>DF$. Since $\frac{AF}{BE}\times\frac{BG}{CF}\times\frac{CH}{DG}=\frac{AH}{DE}$, $\frac{CH}{DG}=1$ and $\frac{AF}{BE}=1$, we have $AH>DE$. Therefore, $AG+AH+BG+BH>CE+CF+DE+DF$, contradiction! If $BG<CF$, similarly, we get a contradiction. This establishes Proposition~\ref{prop:2by4table}.

\begin{prop}
For an $I\times J$ contingency table for discrete random variables $\tilde{X}$ and $\tilde{Y}$ with cell probabilities $\pi_{i,j}$, we could define a set of $(I-1)(J-1)$ local odds ratios
$$
\beta_{i,j}=\frac{\pi_{i,j}\pi_{i+1,j+1}}{\pi_{i,j+1}\pi_{i+1,j}},\quad i=1,\ldots, I-1,\quad j=1,\ldots, J-1.
$$
Then $\tilde{X}\independent \tilde{Y} \Leftrightarrow \beta_{i,j}=1,$ for $i=1,\ldots, I-1$ and $j=1,\ldots, J-1$.
\label{prop:local_odds_ratio}
\end{prop}

To see this, first note that obviously $\tilde{X}\independent \tilde{Y}$ implies that all the local odds ratios equal one. For the reverse, note that for any two rows $i$ and $i+k$ and two columns $j$ and $j+l$, where $1\leq k\leq I-i$ and $1\leq l\leq J-j$, the corresponding odds ratio could be expressed as the product of a set if local odds ratios:
$$
\frac{\pi_{i,j}\pi_{i+k,j+l}}{\pi_{i,j+l}\pi_{i+k,j}}=\prod\limits^{k-1}_{s=0}\prod\limits^{l-1}_{t=0}\beta_{i+s,j+t}.
$$
Therefore, all the odds ratios are equal to $1$. By similar arguments to the proof of Lemma~\ref{lem:equiv_k_indep}, we know that $\tilde{X} \independent \tilde{Y}$. This establishes Proposition~\ref{prop:local_odds_ratio}.

Putting together the above two propositions, we see that $\mathcal{P}_1$ and $\mathcal{P}_2$ are two sets of conditions that guarantee the independence of a $2\times 4$ contingency table.

We now prove that $\theta(A)=0 \text{ for all } A\in \mathcal{A}^{(k_1-1, k_2-1)}$ implies $X \independentk Y$ by induction.

\begin{enumerate}[(i).]
\item For $k_1=1, k_2=1$, it is easy to check that the result holds.

\item Assume that the result holds for $k_1=n_1, k_2=n_2$, for $n_1\geq 1, n_2\geq 1$, that is, 
\begin{equation}
\begin{aligned}
X \independent_{\!\!\!n_1,n_2} Y \quad\Leftarrow\quad \theta(A)=0\quad \text{for all } A\in \mathcal{A}^{(n_1-1, n_2-1)},
\end{aligned}
\label{eq_1}
\end{equation}
we now prove that the result holds for $k_1=n_1, k_2=n_2+1$. 

For $s\in\{1,2,\ldots,2^{n_1}-1\}, t\in\{1,2,\ldots,2^{n_2}-1\}$, let
$$
\beta^{n_1,n_2}_{s,t}:=\frac{F(I^{n_1}_{s}\times I^{n_2}_{t})F(I^{n_1}_{s+1}\times I^{n_2}_{t+1})}{F(I^{n_1}_{s}\times I^{n_2}_{t+1})F(I^{n_1}_{s+1}\times I^{n_2}_{t})}.
$$
According to Proposition~\ref{prop:local_odds_ratio}, (\ref{eq_1}) is equivalent to 
\begin{equation}
\begin{aligned}
\beta^{n_1,n_2}_{s,t}=1\quad  \Leftarrow \quad\theta(A)=0 \quad \text{ for all } A\in \mathcal{A}^{(n_1-1, n_2-1)},
\end{aligned}
\label{eq_2}
\end{equation}
for $s\in\{1,2,\ldots,2^{n_1}-1\}, t\in\{1,2,\ldots,2^{n_2}-1\}$.


\item Again by Proposition~\ref{prop:local_odds_ratio}, we now need to show that
\begin{equation}
\begin{aligned}
\beta^{n_1,n_2+1}_{i,j}=1 \quad \Leftarrow \theta(A)=0  \text{ for all } A\in \mathcal{A}^{(n_1-1, n_2)},
\end{aligned}
\label{eq_3}
\end{equation}
for $i\in\{1,2,\ldots,2^{n_1}-1\}, j\in\{1,2,\ldots,2^{n_2+1}-1\}$. For each $j$, it falls into one of the two cases:

\begin{enumerate}


\item $j$ is odd. $I_j^{n_2+1}\cup I_{j+1}^{n_2+1}=I^{n_2}_{(j+1)/2}\in\mathcal{I}^{n_2}$.

Note that there is a one-to-one mapping between the set of local odds ratios in this case and $\mathcal{A}^{(n_1-1),n_2}$ ($:=\cup_{n^\prime_1\leq n_1-1, n_2^\prime=n_2}\mathcal{A}^{n^\prime_1,n_2^\prime}$):
$$
\beta^{n_1,n_2+1}_{i,j} \mapsto I^{n_1-1-m(i)}_{[(i/2^{m(i)})+1]/2}\times I^{n_2}_{(j+1)/2}\in\mathcal{A}^{(n_1-1),n_2}
$$
where $0\leq m(i)\leq n_1-1$ satisfies that $i/2^{m(i)}$ is odd.

\begin{enumerate}[(1).]
\item
If $m(i)=0$, we have $\beta^{n_1,n_2+1}_{i,j}=1$ by (\ref{eq_3}).
\item

Suppose that $\beta^{n_1,n_2+1}_{i,j}=1$ for all $i$ such that $m(i)\leq m$ and some $m\leq n_1-2$.

Then for $i$ such that $m(i)=m+1$, consider $I^{n_1}_{i^\prime}$, where $i-2^{m(i)}+1\leq i^\prime\leq i+2^{m(i)}-1$, it is obvious that $m(i^\prime)\leq m$. Therefore, we have $\beta^{n_1,n_2+1}_{i^\prime,j}=1$. Consider the $2^{m(i)}\times 2$ table formed by $I^{n_1}_{i^\prime}$ with $i-2^{m(i)}+1\leq i^\prime\leq i$ in the $X$ dimension as well as $I^{n_2+1}_{j}$ and $I^{n_2+1}_{j+1}$ in the $Y$ dimension. In this table, all the local odds ratios are $1$. Using Proposition~\ref{prop:local_odds_ratio}, we have 
\begin{equation}
\begin{aligned}
F(I^{n_1}_{i^\prime}\times I^{n_2+1}_{j})F(I^{n_1}_{i}\times I^{n_2+1}_{j+1})=F(I^{n_1}_{i^\prime}\times I^{n_2+1}_{j+1})F(I^{n_1}_{i}\times I^{n_2+1}_{j}).
\end{aligned}
\label{eq_4}
\end{equation}

Summing (\ref{eq_4}) over $i^\prime$, we have
$$
F(\cup_{i^\prime< i}I^{n_1}_{i^\prime}\times I^{n_2+1}_{j})F(I^{n_1}_{i}\times I^{n_2+1}_{j+1})=F(\cup_{i^\prime< i}I^{n_1}_{i^\prime}\times I^{n_2+1}_{j+1})F(I^{n_1}_{i}\times I^{n_2+1}_{j}).
$$
Similarly, applying the same argument to the $2^{m(i)}\times 2$ table formed by $I^{n_1}_{i^\prime}$ with $i+1\leq i^\prime\leq i+2^{m(i)}$ in the $X$ dimension as well as $I^{n_2+1}_{j}$ and $I^{n_2+1}_{j+1}$ in the $Y$ dimension, we have
$$
F(\cup_{i^\prime>i}I^{n_1}_{i^\prime}\times I^{n_2+1}_{j})F(I^{n_1}_{i}\times I^{n_2+1}_{j+1})=F(\cup_{i^\prime>i}I^{n_1}_{i^\prime}\times I^{n_2+1}_{j+1})F(I^{n_1}_{i}\times I^{n_2+1}_{j})
$$

Note that $\cup_{i^\prime=i-2^{m(i)}+1}^{i+2^{m(i)}} I^{n_1}_{i^\prime}\in \mathcal{I}^{n_1-1-m(i)}$ and because of (\ref{eq_3}), we have 
$$
F(\cup_{i^\prime\leq i}I^{n_1}_{i^\prime}\times I^{n_2+1}_{j})F(\cup_{i^\prime> i}I^{n_1}_{i^\prime}\times I^{n_2+1}_{j+1})=
F(\cup_{i^\prime\leq i}I^{n_1}_{i^\prime}\times I^{n_2+1}_{j+1})F(\cup_{i^\prime> i}I^{n_1}_{i^\prime}\times I^{n_2+1}_{j}).
$$
According to Proposition~\ref{prop:2by4table}, we have $\beta^{n_1,n_2+1}_{i,j}=1$.

\end{enumerate}


\item $j$ is even. $I_j^{n_2+1}\cup I_{j+1}^{n_2+1}\not\in\mathcal{I}^{n_2}$.

In this case,  $I_{j-1}^{n_2+1}\cup I_{j}^{n_2+1}\in\mathcal{I}^{n_2}$ and $I_{j+1}^{n_2+1}\cup I_{j+2}^{n_2+1}\in\mathcal{I}^{n_2}$. Since $j-1, j+1$ are odd, we have $\beta^{n_1,n_2+1}_{i,j-1}=\beta^{n_1,n_2+1}_{i, j+1 }=1$. On the other hand, according to the induction hypothesis (\ref{eq_2}), $\beta^{n_1,n_2}_{i,j/2}=1$. 
Therefore, using Proposition~\ref{prop:2by4table}, we have $\beta^{n_1,n_2+1}_{i,j}=1.$ (See \ref{app_tab_2} for illustration.)

\begin{table}[h]
\begin{center}
  \begin{tabular}{ | c  | c | c | c |}
    \hline
    $\pi_{i,j-1}$ & $\pi_{i,j}$  & $\pi_{i,j+1}$ & $\pi_{i,j+2}$ \\  \hline
    $\pi_{i+1,j-1}$ & $\pi_{i+1,j}$  & $\pi_{i+1,j+1}$ & $\pi_{i+1,j+2}$\\  \hline
  \end{tabular}
\end{center}
\caption{Case (b).}
\label{app_tab_2}
\end{table}

\end{enumerate}

\item By symmetry, the result holds for $k_1=n_1+1, k_2=n_2+1$ following similar arguments.
\end{enumerate}
\end{proof}

\begin{proof}[Proof of Theorem~\ref{thm:likelihood_factorization}]
For $i=0,1,\ldots,k_1$ and $j=0,1,2,\ldots,k_2$, let $\bn_{i,j}:=\{n(A):A\in \A^{i,j}\}$ be the corresponding $2^i\times 2^j$ contingency table for the $(i,j)$-stratum, and let $\F_{i,j}$ be the $\sigma$-algebra generated by $\bn_{i,j}$. For any sequence $i_1\leq i_2 \leq \cdots$ and $j_1 \leq j_2 \leq \cdots$, $\F_{i_1,j_1}\subset \F_{i_2,j_2}\subset \cdots$ form a filtration. 

The proof can be completed by induction. First, the factorization holds by definition for $k_1=k_2=1$. (Also, it holds trivially whenever $k_1=0$ or $k_2=0$, in which case $p(\bn_{k_1,k_2}\,|\,\bn_{k_1,0},\bn_{0,k_2})=1$.) Now without loss of generality suppose the inductive hypothesis holds for $k_1=i-1$ and $k_2=j$ where $i\geq 2$ and $j\geq 1$. Then for $k_1=i$ and $k_2=j$, we have
\begin{align*}
p(\bn_{i,j}\,|\,\bn_{i,0},\bn_{0,j}) = p(\bn_{i,j}\,|\,\bn_{i-1,j},\bn_{i,0},\bn_{0,j})\cdot p(\bn_{i-1,j}\,|\,\bn_{i,0},\bn_{0,j}).
\end{align*}

First, we claim that
\[
p(\bn_{i-1,j}\,|\,\bn_{i,0},\bn_{0,j})=p(\bn_{i-1,j}\,|\,\bn_{i-1,0},\bn_{0,j}),
\]
which can be seen from the following urn argument. Suppose there are $2^{i-1}$ different colors of balls in an urn and the total number of balls of each color in the urn is known. For balls of each color, we randomly assign them $2^{j}$ different labels with the total number of each label assigned also known. We can for example do that by starting with Label~1, and drawing balls without replacement from the urn and assign them Label~1 until the desired number of Label~1 has been assigned. Then we proceed in the same manner with Label~2 and so on and so forth. Now suppose after the assignments we are given the additional information that the balls are of two different sizes---some are large and others are small, and so there are a total of $2^{i-1}\times 2$ different types of balls in the urn and we know the total number of balls of each size within each color. Now, knowing the size does not affect the distribution of the label assignment for the balls as that information was not used in assigning the labels. From this argument, we see that the above equality holds. (In fact, even if the balls are of a variety of different sizes and we are informed of the exact size of each ball, the distribution of the number of each label-color combination is still the same, implying that in fact $p(\bn_{i-1,j}\,|\,\bn_{i',0},\bn_{0,j})=p(\bn_{i-1,j}\,|\,\bn_{i-1,0},\bn_{0,j})$ for all $i'\geq i$.)

On the other hand, because $\F_{0,j}\subset \F_{i-1,j}$, 
\[
p(\bn_{i,j}\,|\,\bn_{i-1,j},\bn_{i,0},\bn_{0,j}) = p(\bn_{i,j}\,|\,\bn_{i-1,j},\bn_{i,0}),
\]
but then by repeatedly applying the above urn argument, we have 
$p(\bn_{i,j'}\,|\,\bn_{i-1,j},\bn_{i,j'-1})=p(\bn_{i,j'}\,|\,\bn_{i-1,j'},\bn_{i,j'-1})$ for $j'=1,2,\ldots,j$. Therefore we have
\begin{align*}
p(\bn_{i,j}\,|\,\bn_{i-1,j},\bn_{i,0}) &= \prod_{j'=1}^{j} p(\bn_{i,j'}\,|\,\bn_{i-1,j'},\bn_{i,j'-1})\\
&=\prod_{j'=0}^{k_2-1}\prod_{A\in \A^{k_1-1,j'}}p(n(A_{00})\,|\,n(A_{0\cdot}),n(A_{\cdot 0}),n(A)).
\end{align*}
The last equality follows because for any $i,j\geq 0$, the $A$'s in $\A^{i,j}$ are non-overlapping and so $n(A_{00})$ are mutually independent conditional on the corresponding row and column totals of~$A$.

Now by the inductive hypothesis, we have
\begin{align*}
p(\bn_{i-1,j}\,|\,\bn_{i-1,0},\bn_{0,j}) &= \prod_{i'=0}^{i-2}\prod_{j'=0}^{j-1}\prod_{A\in\A^{i',j'}} p(n(A_{00})\,|\,n(A_{0\cdot}),n(A_{\cdot 0}),n(A))\\
&= \prod_{i'=0}^{k_1-2}\prod_{j'=0}^{k_2-1}\prod_{A\in\A^{i',j'}} p(n(A_{00})\,|\,n(A_{0\cdot}),n(A_{\cdot 0}),n(A)).
\end{align*}
Putting the two pieces together, we get
\begin{align*}
&p(\bn_{k_1,k_2}\,|\,\bn_{k_1-1,k_2},\bn_{k_1,k_2-1}) \\
=&\prod_{j'=0}^{k_2-1}\prod_{A\in \A^{k_1-1,j'}}p(n(A_{00})\,|\,n(A_{0\cdot}),n(A_{\cdot 0}),n(A))\cdot \prod_{i'=0}^{k_1-2}\prod_{j'=0}^{k_2-1}\prod_{A\in\A^{i',j'}} p(n(A_{00})\,|\,n(A_{0\cdot}),n(A_{\cdot 0}),n(A))\\
=&\prod_{i'=0}^{k_1-1}\prod_{j'=0}^{k_2-1}\prod_{A\in\A^{i',j'}} p(n(A_{00})\,|\,n(A_{0\cdot}),n(A_{\cdot 0}),n(A)).
\end{align*}
By exactly the same argument, one can show that if the inductive hypothesis holds for $k_1=i\geq 1$ and $k_2=j-1\geq 1$, then it also holds for $k_1=i$ and $k_2=j$. 
Therefore, the inductive hypothesis holds for all $k_1$ and $k_2$. This completes the proof.
\end{proof}

\begin{proof}[Proof of Theorem~\ref{thm:fwer_control}]
Suppose we have a screening rule on each window $A$, denoted as a random variable $S(A)$, such that $S(A)=1$ if $A$ passes the screening and so a Fisher's test is applied on $A$, and if $S(A)=0$, then $A$ fails the screening and no test is carried out. The special case without screening will immediate follow by setting $S(A)\equiv 1$. Suppose for each $A\in\A^{i,j}$, $S(A)$ is measurable w.r.t.\ the $\sigma$-algebra generated by $\bn_{i+1,j}$ and $\bn_{i,j+1}$. Correspondingly, $L(i,j)$ is measurable w.r.t.\ that $\sigma$-algebra as well. Now,
\begin{align*}
{\rm P}(p_{overall} \leq \alpha\,|\,H_0,\bn_{k_1,0},\bn_{0,k_2}) &=  {\rm P}(\min _{r} p_{resol}(r) \leq 1-(1-\alpha)^{1/(M+1)}\,|\,H_0,\bn_{k_1,0},\bn_{0,k_2})\\
&= 1 - {\rm P}(p_{resol}(r) > 1 - (1-\alpha)^{1/(M+1)} \text{ for all $r$}\,|\,H_0,\bn_{k_1,0},\bn_{0,k_2}).
\end{align*}
By Theorem~\ref{thm:likelihood_factorization},
\begin{align*}
&{\rm P}(p_{resol}(r) > 1 - (1-\alpha)^{1/(M+1)} \text{ for all $r$}\,|\,H_0,\bn_{k_1,0},\bn_{0,k_2})\\
=&{\rm E}\left(\! \prod_{r:T(r)>0}\!\!\!\! {\rm P}\left(p_{resol}(r) \!>\! 1\!-\!(1-\alpha)^{\frac{1}{M+1}} | H_0,\bn_{k_1,0},\bn_{0,k_2},\{\bn_{i,j}:i+j = r+1\} \right)\Big | H_0,\bn_{k_1,0},\bn_{0,k_2}\!\!\right)\\
=& {\rm E}\left(  \prod_{r:T(r)>0} {\rm P}\left(p_{resol}(r) > 1-(1-\alpha)^{\frac{1}{M+1}}\,|\, H_0,\{\bn_{i,j}:i+j = r+1\}  \right) \,\Big |\, H_0,\bn_{k_1,0},\bn_{0,k_2}\right).
\end{align*}
Now for each $r$ such that $T(r)>0$,
\begin{align*}
&{\rm P}\left(p_{resol}(r) > 1-(1-\alpha)^{1/(M+1)}\,|\, H_0,\{\bn_{i,j}:i+j = r+1\}  \right)\\
=&{\rm P}\left(\min_{i,j:i+j=r,L(i,j)>0} p(i,j) > 1-(1-\alpha)^{1/(M+1)\cdot 1/T(r)}\,|\, H_0,\{\bn_{i',j'}:i'+j' = r+1\} \right)\\
=&\prod_{i,j:i+j=r,L(i,j)>0} {\rm P}\left(p(i,j) > 1-(1-\alpha)^{1/(M+1)\cdot 1/T(r)}\,|\, H_0,\bn_{i+1,j},\bn_{i,j+1} \right)\\
=&\prod_{i,j:i+j=r,L(i,j)>0}\prod_{A\in \A^{i,j}, S(A)=1} {\rm P}\left(p(A) > 1-(1-\alpha)^{1/(M+1)\cdot 1/T(r)\cdot 1/L(i,j)}\,|\, H_0,\bn_{i+1,j},\bn_{i,j+1}\right)\\
\geq& \prod_{i,j:i+j=r,L(i,j)>0}\prod_{A\in \A^{i,j}, S(A)=1} (1-\alpha)^{1/(M+1)\cdot 1/T(r)\cdot 1/L(i,j)}\\
=&(1-\alpha)^{1/(M+1)}.
\end{align*}
Hence 
\begin{align*}
{\rm P}(p_{overall} \leq \alpha\,|\,H_0,\bn_{k_1,0},\bn_{0,k_2}) &\leq {\rm E}\left(1 - (1-\alpha)^{1/(M+1)\cdot |\{r: T(r)>0\}|}\,|\,H_0,\bn_{k_1,0},\bn_{0,k_2}\right)\\ 
&\leq 1- (1-\alpha)^{1/(M+1)\cdot (M+1)}=\alpha.
\end{align*}
\end{proof}

\begin{proof}[Proof of Theorem~\ref{thm:local_consistency}]
Let us focus attention on one $A$ such that $F(A_{\cdot 0}),F(A_{\cdot 1}),F(A_{0\cdot}),F(A_{1\cdot})>0$. In the following,
let ${\rm E}_{\theta}[n(A_{00})\,|\,n(A_{0\cdot}),n(A_{\cdot 0}),n(A)]$ and ${\rm Var}_{\theta}[n(A_{00})\,|\,n(A_{0\cdot}),n(A_{\cdot 0}),n(A)]$ denote respectively the expectation and variance of $n(A_{00})$ given $n(A_{0\cdot}),n(A_{\cdot 0})$, and $n(A)$ when $\theta(A)=\theta$. Then, for any $A$ such that $F(A_{\cdot 0}),F(A_{\cdot 1}),F(A_{0\cdot}),F(A_{1\cdot})>0$, we have that 
${\rm Var}_{\theta}[n(A_{00})\,|\,n(A_{0\cdot}),n(A_{\cdot 0}),n(A)]\rightarrow\infty$
with $F^{\infty}$ probability 1, because with $F^{\infty}$ probability 1, $n(A_{0\cdot})n(A_{1\cdot})n(A_{\cdot 0})n(A_{\cdot 1})/n(A)^3\rightarrow \infty$ \citep{kou:1996}. Now by Theorem~2.2 in \cite{kou:1996}, we have that given $n(A_{0\cdot}),n(A_{\cdot 0}),n(A)$,
\[
Z_{n,\theta}(A) =\frac{n(A_{00})-{\rm E}_{\theta} [n(A_{00})\,|\,n(A_{0\cdot}),n(A_{\cdot 0}),n(A)]}{{\rm Var}^{1/2}_{\theta}[n(A_{00})\,|\,n(A_{0\cdot}),n(A_{\cdot 0}),n(A)]} \rightarrow_{\mathcal{L}} {\rm N}(0,1).
\]
Now,
\begin{align*}
&{\rm P}(p(A) < \alpha(A)\,|\,\theta(A)=\theta,\bn_{k_1,0},\bn_{0,k_2})\\ 
=&  {\rm P}\left(Z_{n,0}(A) < F_{A,n}^{-1}(\alpha(A)/2)\,|\,\theta(A)=\theta,\bn_{k_1,0},\bn_{0,k_2}\right)\\
&\hspace{7em} +{\rm P}\left(Z_{n,0}(A) > F_{A,n}^{-1}(1-\alpha(A)/2)\,|\,\theta(A)=\theta,\bn_{k_1,0},\bn_{0,k_2}\right)\\
\geq & {\rm P}\left(Z_{n,0}(A) > F_{A,n}^{-1}(1-\alpha(A)/2)\,|\,\theta(A)=\theta,\bn_{k_1,0},\bn_{0,k_2}\right)
\end{align*}
where $F_{A,n}$ denotes the exact cdf of $Z_n$ given the marginal totals with $\theta(A)=0$, i.e., the (central) hypergeometric distribution. 

Now, without loss of generality, let us assume that $\theta(A)=\theta >0$. 
By the normal approximation to the hypergeometric distribution we have that
\begin{align*}
&\lim_{n} {\rm P}\left(Z_{n,0}(A) > F_{A,n}^{-1}\left(1-\alpha(A)/2\right)\,|\,\theta(A)=\theta,\bn_{k_1,0},\bn_{0,k_2}\right)\\
=& \lim_{n} {\rm P}\left(Z_{n,\theta}(A) > C_{n}F_{A,n}^{-1}(1-\alpha(A)/2) - B_n C_n\,\Big|\,\theta(A)=\theta,\bn_{k_1,0},\bn_{0,k_2}\right)
\end{align*}
where 
\[
C_{n}=\frac{{\rm Var}^{1/2}_{0}[n(A_{00})\,|\,n(A_{0\cdot}),n(A_{\cdot 0}),n(A)]}{{\rm Var}^{1/2}_{\theta}[n(A_{00})\,|\,n(A_{0\cdot}),n(A_{\cdot 0}),n(A)]}\]
and
\[
B_{n}=\frac{{\rm E}_{\theta} [n(A_{00})\,|\,n(A_{0\cdot}),n(A_{\cdot 0}),n(A)] - {\rm E}_{0} [n(A_{00})\,|\,n(A_{0\cdot}),n(A_{\cdot 0}),n(A)]}{{\rm Var}^{1/2}_{0}[n(A_{00})\,|\,n(A_{0\cdot}),n(A_{\cdot 0}),n(A)]}. 
\]
Now $e^{-\theta/2} \leq C_n \leq e^{\theta/2}$ for all $n$ \cite[Corollary~2.1]{kou:1996}, while $B_n \asymp \sqrt{n}$ with $F^{\infty}$ probability 1. Accordingly, if $k_1$, $k_2$, and $M$ are fixed and thus $\alpha(A)$ is also fixed, i.e., not changing with $n$, then $F_{A,n}^{-1}(1-\alpha(A)/2)\rightarrow \Phi^{-1}(1-\alpha(A)/2)$ and thus $C_{n}F_{A,n}^{-1}(1-\alpha(A)/2) - B_n C_n\rightarrow -\infty$ with $F^{\infty}$ probability 1. Therefore, with $F^{\infty}$ probability 1,
\[
\lim_{n} {\rm P}(p(A) < \alpha(A)\,|\,\bn_{k_1,0},\bn_{0,k_2}) = 1.
\]
Now, in the case when $k_1$, $k_2$, and $M$ can depend on $n$, let $\alpha_n(A)$ be the corresponding window-specific threshold, which is $O(1/\log(n))$. To establish the consistency as above, since $C_n B_n=O(\sqrt{n})$ with $F^{\infty}$ probability 1, we just need to show that $C_n F_{A,n}^{-1}(1-\alpha(A)/2)$ is $o(\sqrt{n})$ with $F^{\infty}$ probability 1. To this end, note that by a Berry-Essen theorem for hypergeometric distributions \cite[Theorem 2.3]{kou:1996}, we have that with $F^{\infty}$ probability 1,
\[
|\Phi(F_{A,n}^{-1}(1-\alpha_n(A)/2) - (1-\alpha_n(A)/2)| < \gamma/\sqrt{n}
\]
for some positive constant $\gamma$ and all large enough $n$. Thus
\[
\Phi^{-1}(1-\alpha_n(A)/2-\gamma/\sqrt{n}) < F_{A,n}^{-1}(1-\alpha_n(A)/2) < \Phi^{-1}(1-\alpha_n(A)/2 + \gamma/\sqrt{n})
\]
Because $1-\Phi(x) \asymp e^{-x^2/2}/x$ as $x\rightarrow \infty$, for  
$\alpha_n(A)=O(1/\log(n))$, $|\Phi^{-1}(1-\alpha_n(A)/2-\gamma/\sqrt{n})-\Phi^{-1}(1-\alpha_n(A)/2)|\rightarrow 0$ and $|\Phi^{-1}(1-\alpha_n(A)/2 + \gamma/\sqrt{n})-\Phi^{-1}(1-\alpha_n(A)/2)|\rightarrow 0$. Hence,
\[
|F_{A,n}^{-1}(1-\alpha_n(A)/2)-\Phi^{-1}(1-\alpha_n(A)/2)| \rightarrow 0.
\]
On the other hand, because $1-\Phi(x) \asymp e^{-x^2/2}/x$ as $x\rightarrow \infty$, we have that $\Phi^{-1}(1-\alpha_n(A)/2)=o(\sqrt{n})$ for $\alpha_n(A)=O(1/\log(n))$. Putting the pieces together, we have with $F^{\infty}$ probability~1,
\[
C_n F_{A,n}^{-1}(1-\alpha(A)/2) = C_n\Phi(1-\alpha_n(A)/2)+ C_n\left(F_{A,n}^{-1}(1-\alpha_n(A)/2)-\Phi^{-1}(1-\alpha_n(A)/2)\right) = o(\sqrt{n}).
\]
This completes the proof.
\end{proof}

\begin{proof}[Proof of Theorem~\ref{thm:global_consistency}]
This theorem follows immediately from the previous one because $p_{overall} < \alpha$ when $p(A)<\alpha(A)$ on any $A$.
\end{proof}

\subsection*{S2.~Additional simulation results}
\subsubsection*{Power study with varying sample size}
In Section~\ref{sec:simulations}, we carried out simulation studies under six different dependency scenarios at fixed sample sizes and varying noise level. Here we carried a power study under the same six dependency scenarios but now with fixed noise levels and varying sample size. The simulation setup is exactly the same as before. We complete 10,000 simulations for each scenario at 20 different sample sizes. For all but the local scenario, the sample sizes range from 50 to 1,000 in increments of 50, while for the local scenario, the sample size ranges from 100 to 2,000 in increments of 100. For each scenario, the noise level is fixed at a particular level that makes the resulting power curve informative. Specifically, the local scenario, the noise level $l=10$ and for the other five scenarios the noise level $l=7$, where the noise level $l$ is defined as in \ref{tab:simulation_setting}. \ref{fig:power_varying_sample_size} presents the power curves of the seven methods.
\begin{figure}[p]
\begin{center}
\includegraphics[width=38em]{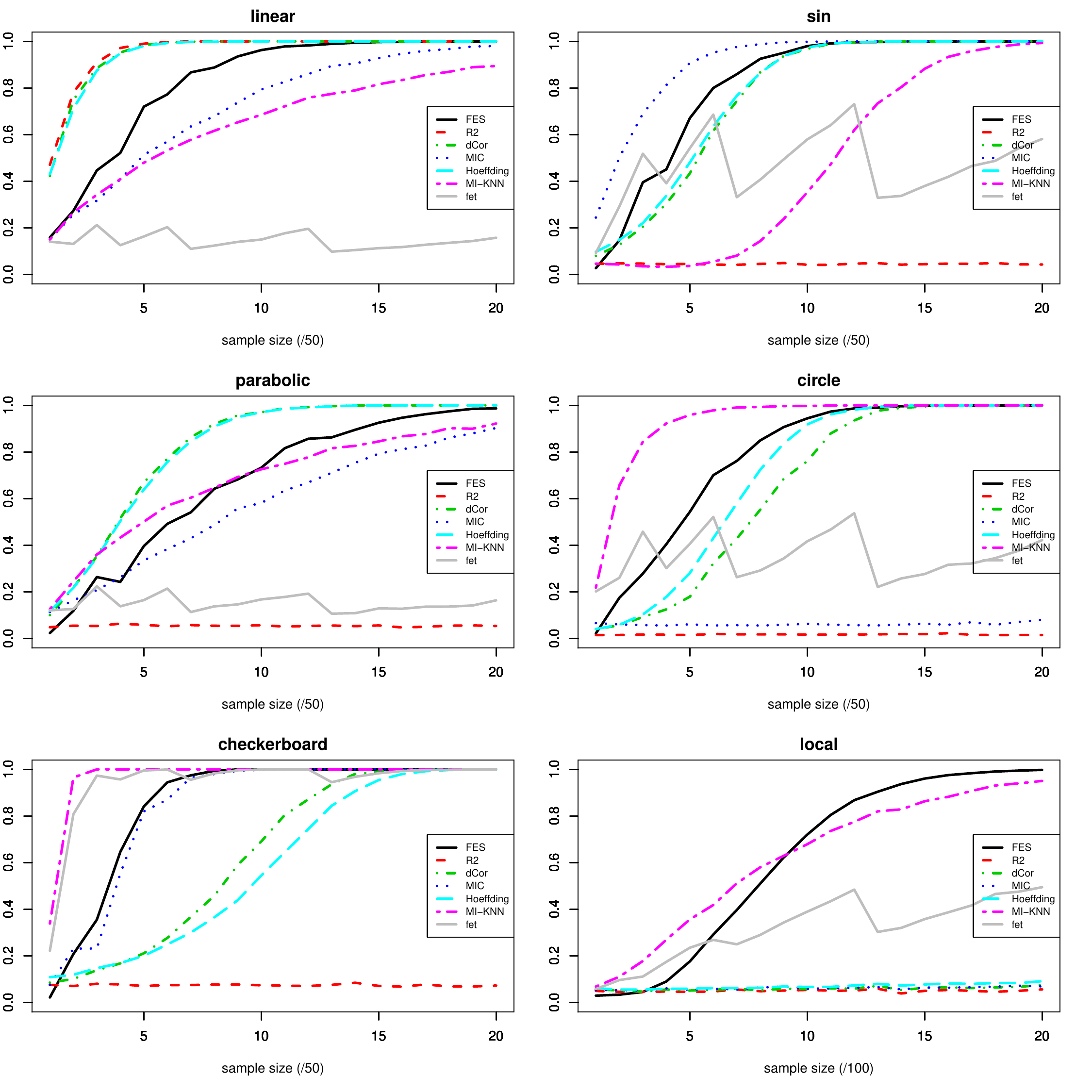}
\caption{Power under the six scenarios at 20 different sample sizes. Seven methods are compared: FES using three-stage exact \v{S}id\'{a}k's correction, Pearson's correlation ($R^2$), distance correlation (dCor), maximal information coefficient (MIC), Hoeffding's $D$ test, $k$-nearest neighbor based mutual information (MI-KNN with $k=10$), and the $R\times C$ Fisher exact test (fet). The significance thresholds for all methods except FES, Hoeffding's $D$, and the classical Fisher's test are computed through permutation. That for the classical Fisher's exact test is computed through standard Monte Carlo.}
\label{fig:power_varying_sample_size}
\end{center}
\vspace{-1em}
\end{figure}
While the results are mostly consistent with the power study with varying noise level and fixed sample size, we note two observations. First, at very small sample sizes, the discreteness of FES does result in a loss of power in comparison to the other methods. Second, a very interesting (and undesirable) feature of the classical $R\times C$ Fisher's exact test (fet) is that its power is not monotonically increasing in sample size, but can display an oscillating pattern. We believe this inconsistency in the performance might be due to the fact that the rejection region of this test is defined as all of the tables with the same marginal totals that have no larger multivariate hypergeometric pmf than the observed one, and hence this rejection region varies as new observations arrive, which alter the marginal totals, in a way not consistent with certain alternatives. We note that through dividing the multivariate hypergeometric into multiple univariate hypergeometric, FES avoids this difficulty. 

\subsubsection*{Sensitivity to choice of $k_1,k_2,M$}
Our next set of simulations investigate the effect of different choices of the resolution parameters $k_1$, $k_2$, $M$ on the power of FES. To this end, we repeat the same simulations as done in Section~\ref{sec:simulations}, but this time, we apply FES at three different choices resolution levels---(i) the recommended resolution level $k_1=k_2=M+1=\lfloor \log(n/10) \rfloor$; (ii) the one level coarse (``-1'') specification $k_1=k_2=M+1=\lfloor \log(n/10) \rfloor -1$; and (iii) the one level finer (``+1'') specification $k_1=k_2=M+1=\lfloor \log(n/10) \rfloor +1$. \ref{fig:power_resolution_perturb} presents the power of FES under the three different resolution choices.
\begin{figure}[p]
\begin{center}
\includegraphics[width=38em]{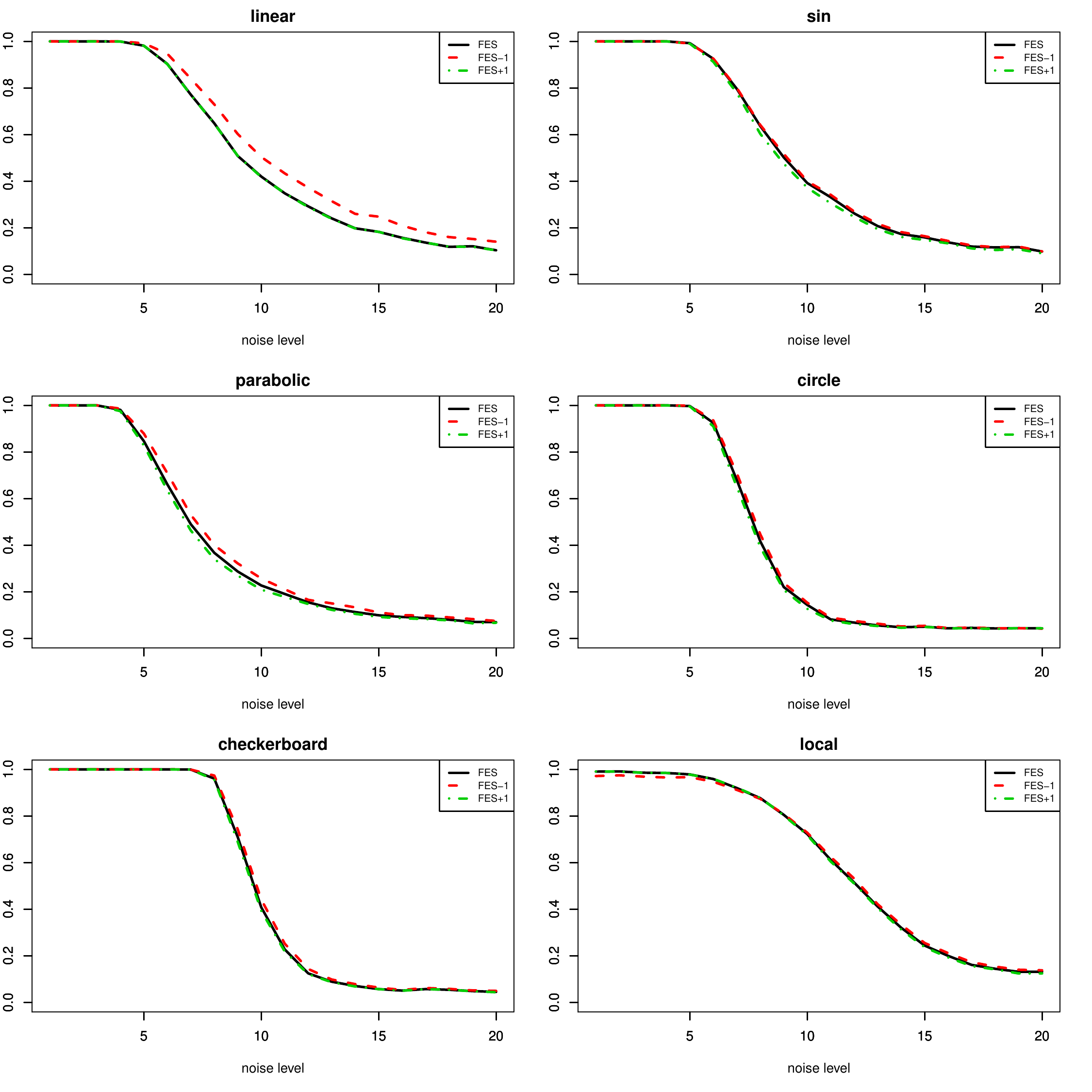}
\caption{Power of FES under the six scenarios at 20 noise levels at three resolution specifications---FES: $k_1=k_2=M+1=\lfloor \log(n/10) \rfloor$; FES$-1$:$k_1=k_2=M+1=\lfloor \log(n/10) \rfloor-1$; FES$+1$:$k_1=k_2=M+1=\lfloor \log(n/10) \rfloor+1$.}
\label{fig:power_resolution_perturb}
\end{center}
\vspace{-1em}
\end{figure}
The results show that FES is generally robust to the choice of resolution levels. Lower resolution levels do tend to result in higher power, especially when the dependency structure is of a large global scale (i.e., affects large portions of the sample space), while higher resolution incur some multiple testing penalty. We also investigated the FWER of FES different resolution specifications through the same null simulation at 20 different sample sizes as done in Section~\ref{sec:simulations}, and \ref{fig:null_power_resolution_perturb} presents the estimated FWER.
\begin{figure}[t]
\begin{center}
\includegraphics[clip=TRUE, trim=0 0 0 8mm, width=26em]{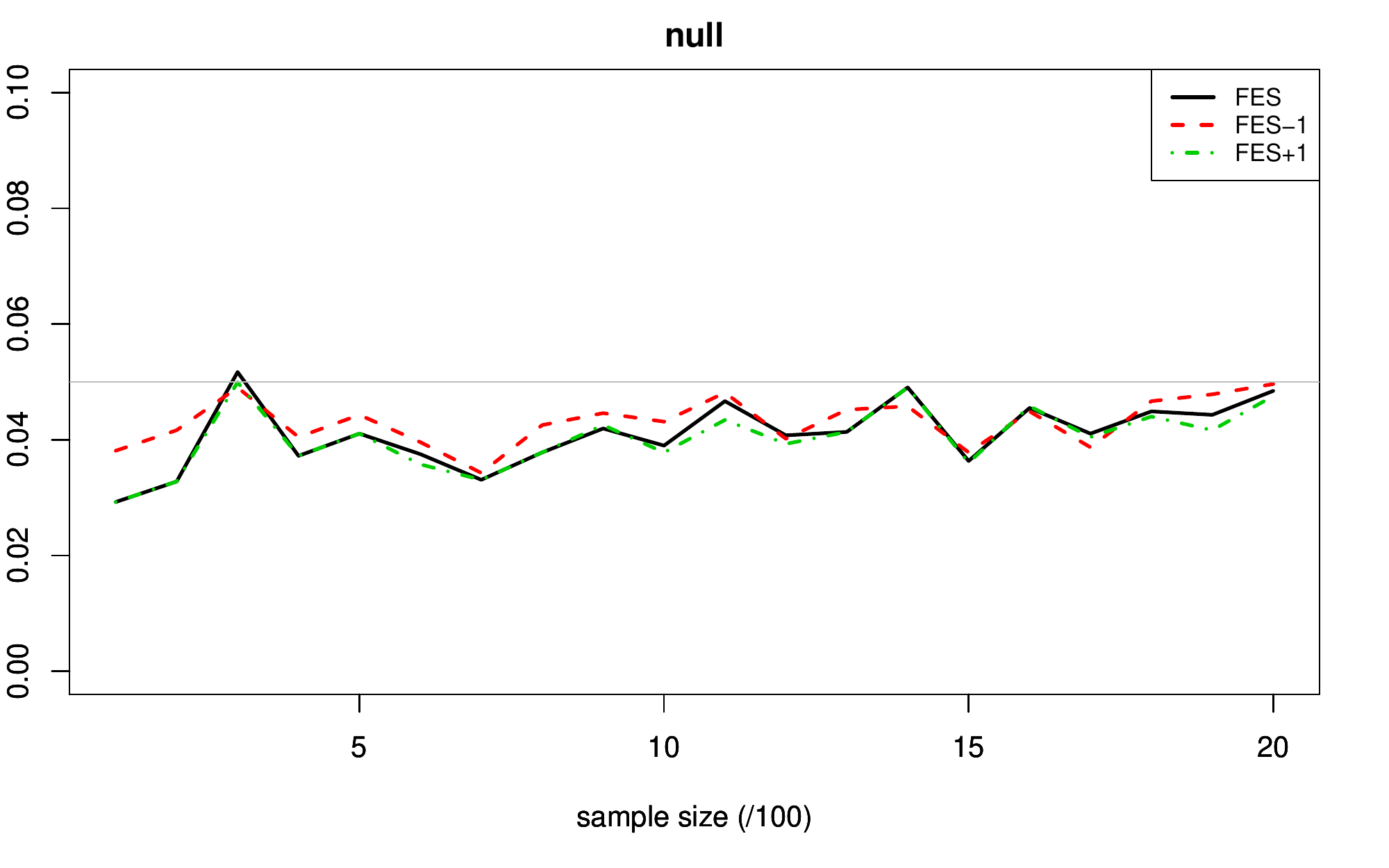}
\caption{FWER of FES under 20 sample sizes at three resolution specifications---FES: $k_1=k_2=M+1=\lfloor \log(n/10) \rfloor$; FES$-1$:$k_1=k_2=M+1=\lfloor \log(n/10) \rfloor-1$; FES$+1$:$k_1=k_2=M+1=\lfloor \log(n/10) \rfloor+1$.}
\label{fig:null_power_resolution_perturb}
\end{center}
\vspace{-1em}
\end{figure}
The results show that adopting higher resolution levels generally has only mild effects on the FWER, but for very small sample sizes, the resulting discreteness in higher resolutions results in additional conservativeness, suggesting that when sample sizes are very small, it is reasonable to only scan at the very coarsest resolution levels.

\end{document}